%% file: main.tex
\documentclass[acmsmall, screen, nonacm]{acmart}

\AtBeginDocument{%
  }

\setcopyright{acmcopyright}
\copyrightyear{2018}
\acmYear{2018}
\acmDOI{XXXXXXX.XXXXXXX}

\acmConference[Conference acronym 'XX]{Make sure to enter the correct
  conference title from your rights confirmation emai}{June 03--05,
  2018}{Woodstock, NY}
\acmPrice{15.00}
\acmISBN{978-1-4503-XXXX-X/18/06}

\usepackage{multirow}
\usepackage{makecell}

\usepackage{enumitem}
\usepackage{stmaryrd}

\usepackage{mathpartir, proof}
\usepackage{listings}
\usepackage{subcaption}
\definecolor{cadmiumgreen}{rgb}{0.0, 0.42, 0.24}

\input{macros}

\begin{document}

\title{FO-Complete Program Verification for Heap 
 Logics}

\author{Adithya Murali}
\orcid{0000-0002-6311-1467}
\email{adithya5@illinois.edu}
\affiliation{%
  \institution{University of Wisconsin–Madison}
  \city{Madison}
  \state{Wisconsin}
  \country{USA}
}

\author{Hrishikesh Balakrishnan}
\email{hb19@illinois.edu}
\orcid{}
\affiliation{
  \institution{University of Illinois Urbana-Champaign, Department of Computer Science}
  \city{Urbana}
  \state{Illinois}
  \country{USA}
}

\author{Aaron Councilman}
\email{aaronjc4@illinois.edu}
\orcid{}
\affiliation{
  \institution{University of Illinois Urbana-Champaign, Department of Computer Science}
  \city{Urbana}
  \state{Illinois}
  \country{USA}
}

\author{P. Madhusudan}
\email{madhu@illinois.edu}
\orcid{0000-0002-9782-721X}
\affiliation{
  \institution{University of Illinois Urbana-Champaign, Department of Computer Science}
  \city{Urbana}
  \state{Illinois}
  \country{USA}
}

\renewcommand{\shortauthors}{Trovato et al.}

\begin{abstract}
We develop the first two heap logics that have implicit heaplets and that admit  FO-complete program verification. The notion of FO-completeness is a theoretical guarantee that all theorems that are valid when recursive definitions are interpreted as fixpoint definitions (instead of least fixpoint) are guaranteed to be eventually proven by the system. The logics we develop are a frame logic (\FL) and a separation logic (\SLFL) that has an alternate semantics inspired by frame logic. We show verification condition generation for FL that is amenable to FO-complete reasoning using  quantifier instantiation and SMT solvers. We show \SLFL\ can be translated to FL in order to obtain FO-complete reasoning. We implement tools that realize our technique and show the expressiveness of our logics and the efficacy of the verification technique on a suite of benchmarks that manipulate data structures.
\end{abstract}

\maketitle

\input{introduction}

\input{sourcesofinc}
\input{prelim}
\input{seplogic}

\input{vcgen}

\input{evaluation}

\input{relatedwork}

\input{dataavailability}

\bibliographystyle{ACM-Reference-Format}
\bibliography{refs}

\newpage
\appendix

\input{appendix.tex}

\end{document}

%% file: macros.tex
\newcommand{\carrow}[1]{~\text{$-$}\!{#1}\!\!\!\shortrightarrow} 
\newcommand{\dom}{\textit{dom}}
\newcommand{\Supp}{\mathit{Supp}}
\newcommand{\FL}{$\textit{FL}$}
\newcommand{\SLFL}{$\textit{SL-FL}$}
\newcommand{\SLFLb}{$\textit{SL-FL}_b$}
\newcommand{\SLFLbd}{$\textit{SL-FL}_b^\vee$}
\newcommand{\proj}{\downharpoonright}

\renewcommand{\star}{\ast}

\newcommand{\Sp}{\mathit{Sp}}
\newcommand{\Cl}[1]{\lbrack#1\rbrack}

\newcommand{\Ll}{\mathcal{L}}
\newcommand{\Cc}{\mathcal{C}}
\newcommand{\Ss}{\mathcal{S}}
\newcommand{\Ff}{\mathcal{F}}
\newcommand{\Dd}{\mathcal{D}}
\newcommand{\Rr}{\mathcal{R}}
\newcommand{\Ii}{\mathcal{I}}

\newcommand{\Loc}{\mathit{Loc}}
\newcommand{\lfp}{\mathit{lfp}}
\newcommand{\nxt}{\mathit{next}}
\newcommand{\key}{\mathit{key}}
\newcommand{\ite}{\mathit{ite}}
\newcommand{\ret}{\mathit{ret}}
\newcommand{\Old}{\mathit{Old}}
\newcommand{\tmp}{\mathit{tmp}}

\newcommand{\nplogic}{L_\mathit{oneway}}
\newcommand{\bfnplogic}{\mathbf{\mathit{L}}_\mathit{\mathbf{oneway}}}

\newcommand{\mypara}[1]{\smallskip\noindent\emph{\textbf{#1}\ }}

\newcommand{\sem}[1]{\llbracket \!#1\! \rrbracket}

\newcommand{\Keys}{\mathit{Keys}}
\newcommand{\lst}{\mathit{List}}
\newcommand{\splst}{\mathit{SpList}}
\newcommand{\nil}{\mathit{nil}}

\newcommand{\weakconj}{\textit{$\wedge\!\!\!\wedge$}}

\newcommand{\opsem}[2]{<#1,#2>}
\newcommand{\redto}{\Downarrow}
\newcommand{\evalto}{\Downarrow_{e}}
\newcommand{\defaultf}{\mathit{default}_{f}}
\newcommand{\assume}{\mathit{assume}}
\newcommand{\alloc}{\mathit{alloc}}
\newcommand{\free}{\mathit{free}}
\newcommand{\iteprog}[3]{\text{if } #1 \text{ then } #2 \text{ else } #3}
\newcommand{\domain}{\mathit{dom}}

\newcommand{\vc}[2]{\text{VC}(#1, #2)}
\newcommand{\Fr}{\mathit{Fr}}
\newcommand{\fr}{\mathit{fr}}
\newcommand{\RD}{\mathit{RD}}

\newcommand{\egprog}{RevPrepend}

\definecolor{dartmouthgreen}{rgb}{0.05, 0.5, 0.06}
\definecolor{indiagreen}{rgb}{0.07, 0.53, 0.03}
\newcommand{\graytext}[1]{\textcolor{indiagreen}{#1}}

\newcommand{\marg}{\mathit{arg}}

%% file: introduction.tex
\section{Introduction}
\label{sec:intro}

Automated verification of programs that destructively manipulate heaps remains a challenging open problem. 
Separation logic~\cite{reynolds02,ohearn04,ohearn01} has emerged as a popular specification logic for expressing properties of structures in heaps. While separation logic is used extensively in interactive theorem proving settings, robust automation of separation logic reasoning has not yet met the same level of success~\cite{vipercav24}. 

One problem for automation of separation logic is that several aspects of it are inherently \emph{second-order}, making it extremely hard to automatically reason with them using SMT solvers (especially SMT solvers on decidable logics). In particular, the \emph{magic wand} ($-\ast$) in separation logic quantifies over \emph{arbitrary heaps} that satisfy a property, which is inherently second-order~\cite{magic-wand-complexity}. 
Separation logic reasoning tools evaluated in competitions (like {\sc SLComp}~\cite{slcomp}) seldom support the magic wand operator and attest to the difficulty of reasoning with it.

Even when the magic wand is avoided (like using incomplete strongest post calculations, e.g.,~\cite{BerdineCalcagnoOHearn2005}), expressing properties such as \emph{separating conjunction}, \emph{conjunction}, and \emph{disjunction} requires quantifying over sets of locations, and when proving implication or entailment, results in formulas with both existential and universal quantification over sets. These are hard to reason with using first-order solvers and, as we show in this paper, results in incompleteness.

The current state-of-the-art for automatic program verification techniques for any expressive heap logic is hence always \emph{incomplete}--- i.e., validity of verification conditions are not just undecidable, but not even recursively enumerable. 
(There are some heap logics that are \emph{decidable}, but these are often heavily restricted in expressiveness~\cite{PiskacWiesZufferey2013, PiskacWiesZufferey2014, PiskacWiesZufferey2014Tool, twbcade13}.). In other words, no complete proof system exists for them and no sound algorithmic procedure can exist that, on all valid verification conditions, halts and declares it valid (even if it's allowed to be non-terminating on invalid ones). 
Systems such as Viper~\cite{vipertool} use several heuristics for magic wand and generate verification conditions that have quantifiers over sets and maps, along with triggers, that are sent to an SMT solver that is inherently incomplete for such logics~\cite{viper-vcgen-technique}. Furthermore, incompleteness manifests in practice as well (see the recent paper~\cite{vipercav24} that studies incompleteness of various algorithms for separation logic using the Viper framework). 

Incompleteness stems from various sources, and we formally identify several sources later in this paper. First, intrinsic second-order quantifications lead to incompleteness, especially quantification over sets. Second, heap logics that are defined using \emph{finite models} for heaps quickly become incomplete, even for first-order logics, due to fundamental incompleteness theorems in finite model theory~\cite{libkin04}. Third, heap logics typically require \emph{recursive definitions} that have a least fixpoint semantics, which causes incompleteness even when models are infinite, as they allow encoding addition and multiplication of integers, and incompleteness follows from G\"odel's incompleteness theorem~\cite{enderton}. Fourth, heap logics for programs need to allow a combined theory of background sorts with models such as linear arithmetic in order to define typical data structures (like a binary search tree), and combination of quantified theories also leads to incompleteness. 

These aspects lead to expressive heap logics being inevitably incomplete. Practical algorithmic techniques for validity for heap logics hence cannot be held to any theoretical standard of effectiveness, such as completeness, and empirical effectiveness on benchmarks are the sole way to judge the efficacy of a technique. Most importantly, we have no theoretical foundation that guarantees that an algorithmic technique will \emph{not} embarrassingly fail to prove a simple verification task, despite being given unbounded resources.

\mypara{The notion of FO-completeness}
In this paper, we pose the audacious question as to whether we can, despite the above challenges, define a notion of completeness that we can hold algorithmic techniques for heap logics up to. The motivation for such a completeness guarantee is that it will yield a pathway to more predictable program verification tools than the current state-of-the-art. We propose \emph{FO-completeness} as a completeness notion for heap logics, inspired by recent results in the literature that show that there are similar domains where FO-complete procedures are both practical and robust.

The single most unavoidable reason for incompleteness of heap logics is the inclusion of recursive definitions. Given a recursive definition of a relation (or function) $f(\overline{x}) =_{\mathit{lfp}} \rho(\overline{x},f)$, $f$ is meant to be interpreted as the \emph{least fixpoint} satisfying the given equation. Such definitions are beyond the power of first-order logics, and moreover render logics incomplete as described above. A simple and practically effective way to bring down this expressive power is to instead interpret recursive definitions as \emph{fixpoint} definitions $f(\overline{x}) = \rho(\overline{x},f)$, where we admit any predicate/function $f$ that satisfies the equation, not just the least fixpoint. 

We say that validity of verification conditions with heap logic specifications is \emph{FO-complete} if they are valid with respect to such a \emph{fixpoint} semantics of recursive definitions in the logic. The name FO-completeness stems from the fact that such equations can be expressed in first-order logic. 
Note that if a theorem is valid with respect to \emph{any} fixpoint that satisfies the recursive definitions, then the theorem is valid with respect to the least fixpoint semantics as well. 

Recent results in the literature have argued that some existing practical, efficient, and robust techniques based on unfolding recursive definitions (``unfold-and-match'' heuristics, "fold/unfold" operations) for program verification are in fact FO-complete. In particular, there is work arguing that the so-called \emph{natural proofs} for certain fragments of first-order logics with recursive definitions have FO-completeness~\cite{loding18}. And there is work arguing that functional program verification techniques over algebraic datatypes (ADTs) used in Liquid Haskell for Haskell programs~\cite{Vazou18} and Leon/Stainless for Scala programs~\cite{BlancKuncakKneussSuter,SuterKoksalKuncak} is also FO-complete~\cite{fluid23}. In the latter setting, in order to express properties in FO logic, the \emph{standard} universe of ADTs is replaced with its FO axiomatization (which admits nonstandard models), and it has been shown that the techniques used in practice are FO-complete under this modeling. 
Note that there is a gap between least fixpoint semantics and fixpoint semantics, and these frameworks require inductive lemmas to bridge this gap. Current tools require users to give these lemmas, but there is active research of synthesizing these lemmas automatically~\cite{fossil, ranjitlemmasynthesis24}.

The above discussion motivates the primary question we pose in this paper:
\begin{quote}
    \emph{Are there expressive heap logics with recursive definitions that have FO-complete validity procedures?}
\end{quote}

More precisely, we would like validity procedures that are sound--- if they report that a theorem is valid, it is valid with respect to the least fixpoint semantics for definitions. And also FO-complete--- if a theorem is valid with fixpoint semantics of recursive definitions, then the procedure is guaranteed to terminate and prove it valid. Note that if a theorem is valid with respect to least fixpoint semantics, but not valid with respect to fixpoint semantics, the tool is not required to terminate.

We are interested in  heap logics that also have the following desirable properties:
\begin{itemize}
 \item The logic should have \emph{implicit heaplets}. Separation logic pioneered the fundamental design principle that every formula $\alpha$ is implicitly associated with a heaplet that captures precisely the subset of locations relevant to $\alpha$, and contracts for function calls use this implicit heaplet as the locations on which permission is taken to modify the heap. We require such implicit heaplets to enable natural frame reasoning by the client, akin to separation logic.
 \item The logic should allow standard combinations with background sorts/theories such as linear arithmetic. 
 \item The logic should be amenable to reasoning using reduction to SMT solvers in order to facilitate scalable and efficient reasoning (note, however, that calling SMT solvers on incomplete logics is a nonstarter). We furthermore require empirical evidence that reasoning is fairly efficient in practice, despite using an FO-complete procedure.
\end{itemize}

\emph{The primary contribution of this paper is to establish the first two heap logics that satisfy the above desiderata and admit FO-complete procedures. We also provide empirical evidence that the FO-complete procedures we develop result in efficient reasoning over a standard benchmark of data structure manipulating programs with rich specifications of their functional behavior. The first heap logic is Frame Logic~\cite{esop2020framelogic, framelogictoplas2023}, and the second heap logic is a new separation logic that has an alternate semantics inspired by frame logic.}

%\medskip
%\noindent {\bf Frame Logic:} 
\mypara{An FO-complete Frame Logic}
Our first focus of investigation is  \emph{Frame Logic} (FL), proposed recently by Murali et al 
% L\"oding et al~\ 
that embraces the \emph{principles} of separation logic but works in a \emph{first-order logic setting with recursive definitions (FORD)}~\cite{framelogictoplas2023,esop2020framelogic}.
Frame logic (FL) is aesthetically simple--- it adds to classical first-order logic with recursive definitions a \emph{support} operator $\Sp$, where $\textit{Sp}(\alpha)$, for any formula $\alpha$, evaluates to the implicit subset of locations of the heap that the truthhood/falsehood of $\alpha$ relies upon. The support is similar to local heaplets in separation logic. However, unlike separation logic, supports are \emph{uniquely} defined, which allows modeling them without true quantification over sets.

Instead of relying on a localized heaplet semantics as separation logic does, the semantics of FL keeps the classical semantics of FORD over the \emph{global heap}, but allows recovering the local heaplet using the support operator. Frame logic can express the separating conjunct 
--- $\alpha * \beta$ is essentially expressed as $\alpha \wedge \beta \wedge \textit{Sp}(\alpha) \cap \textit{Sp}(\beta)=\emptyset$. 

Work on frame logic by Murali et al~\cite{esop2020framelogic, framelogictoplas2023} has argued that weakest preconditions for programs with frame logic annotations are expressible in frame logic itself, and that verification conditions in frame logic can be translated to FORD. 
However, this does not translate to automation (especially using SMT solvers) as FORD with alternating quantifiers is highly complex.
The weakest precondition transformations described in~\cite{esop2020framelogic,framelogictoplas2023} are complex, involving several introductions of \emph{quantifiers}  %(including existential quantification)
leading to weakest preconditions having alternations of universal and existential quantification.
Furthermore, weakest preconditions introduce first-order formulations of magic wand (``MW''-operators) that have gnarly definitions that are too complex to automate.

In this paper, we approach automated verification of frame logic afresh with the goal of embedding verification conditions into an existing automatable class of first-order logic with recursive definitions. The technique of \emph{natural proofs} is a well-established technique for handling such logics--- it uses recursive definition unfoldings, uninterpreted function abstractions, and SMT-based quantifier-free reasoning to build sound automated reasoning~\cite{qiu13,pek14,loding18}.
Furthermore, the work on its foundations~\cite{loding18}
identifies a particular fragment $\nplogic$ for which the natural proofs technique is $FO$-complete.

The first technical contribution of this paper is to build an FO-complete automated program verification paradigm for frame logic specifications that generates verification conditions in the logic $\nplogic$. We overcome several challenges using novel ideas, including (a) defining an adequately powerful fragment of frame logic that allows a restricted form of quantification using a new {\tt cloud} operator, and (b) generating verification conditions with care so that they are in the $\nplogic$ fragment and have existential quantification only over the foreground sort. The verification conditions are generated using strongest-post like symbolic evaluation rather than weakest preconditions, avoiding magic wand like operators.

\mypara{An FO-complete Separation Logic}
The second contribution of this paper is a separation logic with an \emph{alternate semantics} than classical semantics, and an FO-complete procedure for program verification for it. The new logic, \SLFL, is a separation logic that has \emph{determined heaplets} inspired by frame logic. As far as we know, this is the first separation logic that is powerful enough (has recursive definitions and an undecidable validity problem) and that admits $FO$-complete program verification.

The logic \SLFL\ is interpreted over stores and heaplets, $(s,h)$, instead of a global heap, similar to traditional separation logic. However, it is designed so that for any formula $\alpha$ and for any global heap $g$, there is \emph{at most one heaplet} $h$ of $g$ on which $\alpha$ holds. It is in this sense that \SLFL\ has determined heaplets. Note that traditional separation logic does not have this property--- $\alpha \vee \beta$ can hold in multiple heaplets of a global heap, say one satisfying $\alpha$ and one satisfying $\beta$, and formulas like $\mathit{true}$ hold on multiple heaplets. The semantics of \SLFL\ uses auxiliary support operators inspired by frame logic in order to define determined heaplets.

We show that \SLFL\ formulas can be translated to frame logic, where $\alpha$ in \SLFL\ translates to a formula $\alpha'$ in frame logic, where although $\alpha'$ is interpreted on global heaps, its support $\Sp(\alpha')$ is precisely the unique heaplet under which $\alpha$ could hold.
Coupling this translation with the FO-complete verification for programs against frame logic specifications gives us an FO-complete program verification for \SLFL. \SLFL\ has several other features that also facilitate FO-completeness--- specifications are purely quantifier-free, but we allow a set of recursive definitions (that require quantification), and global heaps are interpreted as potentially infinite models, avoiding another source of incompleteness.

\mypara{Evaluation}
The third contribution of this paper is an evaluation that addresses two primary research questions: (a) are the logics we define expressive enough for full functional specification of standard data structure manipulating routines, and (b) are the FO-complete program verification mechanisms efficient in practice in verifying such programs.
We develop tools, a tool FLV (Frame Logic Verifier) for programs annotated with frame logic specifications, and a checker SLFLV that translates \SLFL\ formulae to facilitate FO-complete verification against \SLFL\ specifications. 
%The tool caters to a simple programming language that destructively updates heaps. 
FLV generates verification conditions in first-order logic with recursive definitions, in particular in the logic $\nplogic$, and utilizes an existing tool to do natural proof based FO-complete reasoning~\cite{qiu13,pek14,loding18,fossil}.  
We evaluate our tool on a suite of programs that manipulate common data structures, and find that both logics are expressive enough for stating rich functional specifications for these methods, and facilitate efficient verification.

\medskip
In summary, our work develops the first expressive heap logics, a frame logic and
a separation logic, that have implicit heaplets and admit FO-complete validity procedures. We provide the accompanying program verification methodologies and an evaluation that attests to the expressiveness of the logics and efficiency of the automation. This work hence establishes a new theoretical standard of completeness for heap logics. Our hope is that demanding this new standard of FO-completeness for future heap logics paves the way for more robust and predictable heap verification.

%% file: sourcesofinc.tex
\section{Sources of Incompleteness in Heap Logics}
\label{sec:sourcesofinc}

In this section, we show several results demonstrating key sources of incompleteness in heap logics. 
In particular, we will show (a) restricting to finite models gives incompleteness, and hence we must define models that include infinite models, (b) definitions with least fixpoint semantics cause incompleteness, which motivates our study of completeness of definitions with arbitrary fixpoint semantics (called FO-completeness), and (c) even when models are infinite and we consider fixpoint definitions, classical separation logic makes design decisions related to having multiple heaplets that can satisfy a formula, which intrinsically causes incompleteness.

\paragraph{Completeness and Incompleteness for Logics} We say a logic $\Ll$ is \emph{complete}
if the set of all valid formulas in $\Ll$ is recursively enumerable (written r.e.). In other words, there exists a sound procedure that can prove every valid formula (but may not halt on invalid formulas). 
Conversely, we say that $\Ll$ is \emph{ incomplete} if validity is not r.e.. 
Note that incompleteness implies undecidability, but the converse is not true (e.g., pure FOL is undecidable but complete~\cite{turing1936, enderton}). 
\smallskip

\mypara{Finite Models} A model is said to be finite if the size of its universe is finite. Logics whose semantics is defined only over finite models are typically incomplete. This follows from general results in Finite Model Theory~\cite{libkin04}.

\begin{theorem}[Trakhtenbrot's Theorem]
\label{thm:fol-finite-incomplete}
Fix a first-order signature $\Sigma$, and let $\varphi$ be a formula over that signature. The problem of checking whether $\varphi$ is valid over finite models is not recursively enumerable.
\end{theorem}

The proof can be shown using a reduction from the \emph{non-halting problem} for Turing machines to validity of FOL over finite models. In fact, we can have linked lists encode configuration sequences of Turing machines, and assert that they never encode a halting state. Given that linked lists are the most basic data structures, finite models are a nonstarter for designing complete logics. Separation logics are typically defined over finite heaps~\cite{reynolds02}. We avoid this source of incompleteness by interpreting our logics over finite and infinite models.

\mypara{Recursive Definitions} A second source of incompleteness for heap logics stems from modeling heap data structures and measures as recursive definitions with least fixpoint semantics. This holds even when the logic has both finite and infinite models; the following is well-known:

\begin{proposition}
\label{thm:fo-lfp-incomplete}
First-Order Logic augmented with least fixpoint (lfp) recursive definitions (see Section~\ref{sec:ford} for definition) is incomplete.
\end{proposition}

In fact, one can define a ``number line'' along with addition and multiplication using a \emph{single unary} lfp definition (see~\cite{fossil}), yielding incompleteness by G\"{o}del's incompleteness theorem for arithmetic~\cite{ebbinghausflumthomas-logicbook}. This shows that accomodating lfp definitions are a nonstarter for achieving completeness. In this paper, we hence consider \emph{FO-completeness} which is the notion of completeness where definitions are interpreted using arbitrary fixpoint semantics.

\mypara{Monadic Second-Order Quantification} Yet another source of incompleteness for heap logics stems from design choices that--- either implicitly or explicitly--- allow non-nested monadic second-order quantification (quantification over sets). For example, the separating conjunction operator $*$ in Separation Logic (SL) has the following semantics: ``$s,h \models \alpha * \beta$ if there exist disjoint subsets $S_1,S_2$ of the domain of $h$ such that...''. Consequently, checking an entailment $\alpha * \beta \models \gamma * \delta$ essentially requires reasoning with a formula of the form $(\exists S_1,S_2.\, \mathit{Valid}(S_1,S_2) \land \Phi) \rightarrow (\exists S_3,S_4.\,\mathit{Valid}(S_3,S_4) \land \Phi')$. Note that this formula has both existential and universal quantification over sets (due to the implication).
We provide evidence that even such non-nested second-order quantification prohibits completeness:

\begin{theorem}[Proof in Appendix~\ref{app:sourcesofinc}]
\label{thm:so-high-undec}
Monadic Second-Order Logic is incomplete. This holds even where there is only one level of quantification over sets (no nesting).
\end{theorem}
\noindent
Note that, in the above result, we consider validity over both finite and infinite models. 
We do not know whether the above result is known, and provide a proof in Appendix~\ref{app:sourcesofinc}. 

\smallskip
In this work, we seek to avoid this source of incompleteness by designing logics where formulas have \emph{unique heaplets}. Intuitively, when formulas have unique heaplets, we can write their semantics using universal quantification or existential quantification. After all, the universal quantifier can only be satisfied by the uniquely defined heaplet of the formula! Therefore, instead of checking $(\exists S_1,S_2.\, \mathit{Valid}(S_1,S_2) \land \Phi) \rightarrow (\exists S_3,S_4.\,\mathit{Valid}(S_3,S_4) \land \Phi')$, we can check $\forall S_1, S_2, S_3, S_4.\, \big((\mathit{Valid}(S_1,S_2) \land \Phi) \rightarrow (\mathit{Valid}(S_3,S_4) \rightarrow \Phi')\big)$. 
This formula only contains universal quantification over sets, and can therefore be removed and considered as part of the first-order signature (as a monadic predicate).

\smallskip
We conclude this section with a new incompleteness result for SL. 
Although prior work has identified that operators like the magic wand are inherently second-order~\cite{magic-wand-complexity} (which would yield incompleteness), we show SL is incomplete due to non-unique heaplets even if we eliminate the magic wand as well as other sources of incompleteness such as finite models and least fixpoints:

\begin{theorem}[Non-Unique Heaplets make SL Incomplete]
\label{thm:sl-fp-incomplete}
Separation Logic (a) without the magic wand, (b) augmented with fixpoint definitions, and (c) defined over both finite and infinite heaplets, is incomplete.
\end{theorem}

We formally describe SL extended to infinite heaplets and fixpoint definitions, as well as the proof of the above result in Appendix~\ref{app:sourcesofinc}. Our proof crucially uses the fact that SL formulas do not have unique heaplets.

In view of these results, we design logics in this paper that avoid these pitfalls in our search for FO-complete heap logics. The logics \FL~ and \SLFL we design are interpreted over models of arbitrary size (finite and infinite) and have determined heaplets, and we show both of them to admit FO-complete program verification.

%% file: prelim.tex
\section{Frame Logic and Program Verification}
\label{sec:prelim}

In this section we introduce the first logic for which we design FO-complete reasoning: Frame Logic (FL). FL was introduced in earlier work by Murali et al.~\cite{esop2020framelogic,framelogictoplas2023} as an extension of First-Order Logic with Recursive Definitions (FORD) using a \emph{Support} operator $\Sp$ to access the heaplets of formulas. We present a high level technical summary of the key concepts including the syntax and semantics of FL, as well as the notion of validity for Hoare Triples for which we design automation techniques in this work. The ideas presented in this section are not the contribution of this paper and are developed in prior works~\cite{framelogictoplas2023,loding18}.

\subsection{First-Order Logic with Recursive Definitions}
\label{sec:ford}

We now present First-Order Logic with Recursive Definitions (FORD). FORD is similar to first-order logic with least fixpoints~\cite{libkin04,vardi1982,immerman1982,relational-logic+lfp,chandra-harel}, except recursive definitions (which have least fixpoint semantics) are given names.

Formally, we have a signature $\Sigma = (\Ss, \Cc, \Ff, \Rr, \Ii)$ where $\Ss$ is a finite set of sorts, and $\Cc,\Ff$ and $\Rr$ are sets of constant, function, and relation symbols respectively. $\Ii$ is a set of relation symbols disjoint from $\Rr$ whose interpretations are given using recursive definitions (as opposed to being interpreted by a model). Symbols have their usual types, e.g., function symbols in $\Ff$ have an associated arity $n \in \mathbb{N}$ and are of type $\tau_1 \times \tau_2\ldots \times \tau_n \rightarrow \tau$, where $\tau_i,\tau \in \Ss$.

We require that $\Ss$ contain a designated \emph{foreground} sort $\Loc$. We use the foreground sort to model heap locations. The remaining sorts, called \emph{background} sorts, are used to model data values such as integers, sets, etc. %and they are typically sorts that are used in verification such as integers, sets, etc.
We use the function symbols in $\Ff$ to model pointers and data fields of heap locations. For example, the next pointer of a linked list can be modeled using the symbol $\nxt: \Loc \rightarrow \Loc$, and similarly, the key stored at a location can be modeled using $\key: \Loc \rightarrow \mathit{Int}$.

We do not provide the syntax of FORD here as it is essentially identical to the syntax of Frame Logic given in Figure~\ref{fig:syntax}, except that FORD does not contain the support operator $\Sp(\cdot)$.

\mypara{Recursive Definitions} A recursive definition of a predicate $I \in \Ii$ is of the form $I(\overline{x}) :=_{\mathit{lfp}} \rho(\overline{x})$, 
where $\rho$ is a quantifier-free formula that only mentions recursively defined symbols in $\Ii$ positively (i.e., under an even number of negations). This ensures that least fixpoints always exist~\cite{tarski-knaster}. We formally treat recursively defined functions by modeling them as predicates, however, we will use function symbols with recursive definitions in our exposition. We denote the set of definitions for the symbols in $\Ii$ by $\Dd$. We require that $\Dd$ contains exactly one definition for each $I \in \Ii$.

\mypara{Semantics} We consider first-order models where the foreground sort $\Loc$ is uninterpreted and the background sorts are constrained by a first-order theory. This theory is usually the combination of several theories over individual sorts such that the quantifier-free fragment of the combination is decidable~\cite{nelson-oppen1979}.
Given a set of definitions $\Dd$ for the symbols in $\Ii$, a model of FORD consists of a first-order model of the above kind that interprets the symbols in $\Cc$, $\Ff$, and $\Rr$ (respecting the various theories), as well as an interpretation for the symbols in $\Ii$ that is determined by the first-order model as the least fixpoint of the definitions $\Dd$. %are computed on the first-order model as the least fixpoint of their definitions $\Dd$. 
Formulas are then evaluated as usual.

\mypara{The $\bfnplogic$ Fragment} $\nplogic$ is a syntactic fragment of FORD introduced in prior work~\cite{loding18} which allows only ``one-way" functions from the foreground sort to the background sorts. Formally, every function symbol in $\Ff$ of arity $n$ whose range sort is the foreground sort $\Loc$ has domain $\Loc^n$. Recursively defined symbols $\Ii$ of arity $k$ are of type $\Loc^k$. Finally, formulas are only allowed to quantify existentially over $\Loc$ (for validity). 
Validity checking for formulas in $\nplogic$ is FO-complete and can be automated effectively using a systematic quantifier-instantiation procedure~\cite{loding18} based on Natural Proofs~\cite{qiu13,pek14}.

\subsection{Frame Logic with Guarded Quantification}
\label{sec:fl}

Frame Logic (FL)~\cite{esop2020framelogic,framelogictoplas2023} is a heap logic based on FORD with \emph{implicit} heaplets. Syntactically, FL extends FORD with a \emph{support} operator $\Sp(\cdot)$ that allows access to the implicit heaplet. Formally, given a formula $\varphi$ (term $t$), $\Sp(\varphi)$ denotes the subset of locations (foreground elements) on which the truth of $\varphi$ (resp. value of $t$) depends. %In the context of heaps, the support can be thought of as the ``heaplet'' of $\varphi$. 
One can then use the $\Sp$ operator to state disjointness and reason about framing, e.g., the formula $\lst(x) \land \lst(y) \land \Sp(\lst(x)) \cap \Sp(\lst(y)) = \emptyset$ says that $x$ and $y$ point to disjoint linked lists. % (the corresponding Separation Logic formula would be $\lst(x) * \lst(y)$).

\mypara{Syntax} 
In this work, we use a fragment of FL with guarded quantification shown in Figure~\ref{fig:syntax}. Formally, we distinguish a set of \emph{mutable functions} $\Ff_m$ among the symbols in $\Ff$ with domain $\Loc$ that model pointer and data fields over the foreground sort. We also require a background sort $\mathit{Set}(\Loc)$ representing sets of foreground locations and define $\Sp(\varphi)$ as a term of type $\mathit{Set}(\Loc)$. We utilize in the syntax $\ite$ (``if-then-else'') expressions over terms and formulas, where $\ite(\gamma: \alpha, \beta)$ denotes ``if $\gamma$ holds then $\alpha$ else $\beta$''. Note that the \emph{guard} $\gamma$ cannot mention inductively defined predicates or terms of type $\mathit{Set}(Loc)$, including support expressions. These are technical restrictions that are required to define a well-defined semantics for Frame Logic formulas.

The fragment of FL that we use in this work has two restrictions. First, we only allow \emph{guarded} quantification of the form $\exists y\!: y = f(x).\, \varphi(y)$ where $y$ is a variable over the foreground sort $\Loc$. The truth value of this formula is the same as $\exists y\!: y = f(x) \land \varphi(y)$, but its support is defined more carefully. We describe this below. Second, we restrict the signature of function symbols as well as recursively defined predicates/functions to be ``one-way'' in the sense of the $\nplogic$ fragment of FORD described above. We reduce reasoning about programs with specifications written in this fragment to the validity of formulas in $\nplogic$.

\begin{figure}\small
\[
\begin{array}{rrcl}
\text{FL Formulas} & \varphi & \coloneq & \bot \mid \top \mid t = t
\mid R(t_1, \dots, t_m) \mid \varphi\land\varphi \mid \varphi \lor \varphi \mid \neg \varphi\\
&&& \mid \ite(\gamma:\varphi, \varphi) \mid \exists y\!: y = f(x).\;\varphi\\
&&& \mbox{where } R \in \Rr \cup \Ii,\; y \in \mathit{Var}_\Loc,\; f \in \Ff_m\\
\text{Guards} & \gamma & \coloneq & t = t \mid R(t_1, \dots, t_m) \mid \gamma\land \gamma \mid \gamma \lor \gamma \mid\neg \gamma \mid ite(\gamma:\gamma, \gamma)\\
&&& \mbox{where $R \in \Rr$,\; $t_i,t$ are not of type $\mathit{Set}(\Loc)$}\\
\text{Terms} & t & \coloneq & c \mid x \mid f(t_1, \dots, t_m) \mid ite(\gamma: t, t)\\
% &&& \mid \Sp(\varphi) \mid \Sp(t') \;\;\;\text{if $t$ is of type $\mathit{Set}(\Loc)$}\\
&&& \mid (\Sp(\varphi) \text{ if t is of type $\mathit{Set}(\Loc)$} )\mid ( \Sp(t') \text{ if t is of type $\mathit{Set}(\Loc)$} )\\
&&& \mbox{where $c$ is a constant, $x$ is a variable, and $t_1,\dots,t_m$ are of}\\
&&& \mbox{the appropriate type.}\\
\text{Recursive definitions} & I(\overline{x}) & \coloneq & \rho_{I}(\overline{x}) \text{ where $I\in \Ii$, $\overline{x} \in \Loc^k$ for some $k \in \mathbb{N}$, $\rho_{I}$ is a  }\\
&&& \mbox{ FL-formula where all relation symbols $I'\in \mathcal{I}$ occur}\\
&&&  \mbox{ only positively or inside a support expression}
\end{array}
\]
\vspace{-0.5 cm}
\caption{Frame Logic with guarded quantification. $\mathit{Var}_\Loc$ denotes variables over the foreground sort and $\mathit{Set}(\Loc)$ denotes the sort consisting of sets of foreground elements. Terms and formulas are assumed to be well-typed for simplicity of presentation.
}
\label{fig:syntax}
\vspace{-0.5 cm}
\end{figure}

\mypara{Semantics} FL semantics extends semantics for FORD. We refer the reader to prior work on FL~\cite{framelogictoplas2023} for a detailed presentation. We focus here on the semantics of the $\Sp(\cdot)$ operator. $\Sp$ is defined as a set of recursive equations in Figure~\ref{fig:sp-semantics}. Given a model $M$, the support of a formula $\alpha$ in $M$ (resp. term), denoted $\sem{\Sp(\alpha)}_M$, is the least fixpoint of the equations in Figure~\ref{fig:sp-semantics}.

The support can be understood as the set of locations (i.e., ``heaplet'') on which mutable functions must be applied (i.e., dereferenced) in order to compute the value of a given term or formula. The support of constants is empty. The support of a term $f(t)$ in $M$ for a mutable function $f \in \Ff_m$ is, as expected, $\{\sem{t}_M\}$. %If we think of $f(t)$ as a dereference of $f$ on $t$, then its heaplet is indeed $\{t\}$. 
The application of a non-mutable function does not contribute to the support.%, and therefore the support of a term is the set of all dereferenced subterms. %Terms under non-mutable functions, constants, etc, have an empty support.

The support of $\alpha \land \beta$ is intuitively the union of the supports of $\alpha$ and $\beta$, and this is indeed the case in FL. Note, however, that the support of $\alpha \lor \beta$ is also the union of the supports of $\alpha$ and $\beta$. This is because Frame Logic defines a \emph{unique} support for a formula $\emph{regardless of its truth value}$. We can also see this reflected in the definition of $\Sp(\neg\alpha)$, which is equal to $\Sp(\alpha)$. This is different from, say, Separation Logic~\cite{reynolds02,seplogicprimer,demri15}, where the heaplet of $\alpha \lor \beta$ is the heaplet of any of the disjuncts that evaluate to true (formulas can have multiple heaplets). Unique heaplets are a salient feature of FL, and the logic makes several design decisions to achieve this. We point the reader to prior work~\cite{framelogictoplas2023} for a discussion on the ramifications of these design decisions.

The support of $\ite(\gamma: \alpha, \beta)$ always includes the support of $\gamma$ (since we need to evaluate $\gamma$ to determine the truth of the formula), but then only adds the support of the case that is evaluated depending on $\gamma$. We use the $\ite$ rather than $\lor$ to write expressions with finer-grained supports in FL. The support of an inductively defined relation term $\Sp(I(\overline{x}))$ is simply the support of the body of the definition $\Sp(\rho(\overline{x}))$. The body $\rho$ may of course mention $I$ recursively, and the support is the least fixpoint of these equations.

Finally, the support of $\exists y\!: y = f(x).\; \varphi(y)$ contains the location interpreted by $x$, as well as the support of $\varphi(y)$ where $y$ is interpreted to be location corresponding to $f(x)$ in the given model. Note that although the formula evaluates the same as $\exists y.\, y = f(x) \land \varphi(y)$, its support only includes the support of $\varphi$ for values of $y$ `matching' the guard, namely $f(x)$.

\begin{comment}
\medskip
\textit{Aside: Formal Semantics of FL.} 
Observe from Figure~\ref{fig:sp-semantics} that the definition of $\Sp$ does not depend on the semantics of FL formulas in a mutually recursive manner. The support of an $\ite$ expression does depend on the value of the guard $\gamma$, but since we disallow recursively defined functions and support expressions from guards (see Figure~\ref{fig:syntax}), supports can be defined essentially independently from the semantics of FL formulas. Intuitively, one can therefore think of the semantics of FL formulas in the following way: fix first a First-Order model $M$ that interprets symbols in $\Cc \cup \Ff \cup \Rr$, but does not give valuations to recursively defined functions or supports. We first compute supports for formulas according to Figure~\ref{fig:sp-semantics}. As we argued above, we can do this because the supports of formulas can be defined only using the interpretation for symbols in $\Cc \cup \Ff \cup \Rr$. We then interpret recursively defined symbols in $\Ii$ according to their definitions using least-fixpoint semantics, producing a model $M^\Ii_{\Sp}$. Models of the form $M^\Ii_{\Sp}$ obtained in the above way are called \emph{frame models}. FL formulas are then simply first-order formulas over such frame models, and are evaluated in the usual way. We refer the reader to prior work on Frame Logic for the formal construction of frame models~\cite{framelogictoplas2023}.
\end{comment}

\begin{figure}\small
\begin{minipage}{.5\linewidth}
\begin{align*}%{lcl}
\sem{\Sp(c)}_M  = \sem{\Sp(x)}_M &= \emptyset \mbox{ for constant $c$, variable $x$}\\
\sem{\Sp(\top)}_M &= \sem{\Sp(\bot)}_M = \emptyset\\
\sem{\Sp(\Sp(t))}_M     &= \sem{\Sp(t)}_M\\
\sem{\Sp(t_1 = t_2)}_M &= \sem{\Sp(t_1)}_M \cup \sem{\Sp(t_2)}_M \\
\sem{\Sp(\alpha \land \beta)}_M  &= \sem{\Sp(\alpha)}_M \cup \sem{\Sp(\beta)}_M\\
\sem{\Sp(\alpha \lor \beta)}_M  &= \sem{\Sp(\alpha)}_M \cup \sem{\Sp(\beta)}_M\\
\sem{\Sp(\Sp(\varphi))}_M  &= \sem{\Sp(\varphi)}_M\\
\sem{\Sp(\neg \varphi)}_M &= \sem{\Sp(\varphi)}_M\\
\sem{\Sp(I(t_1\dots, t_n))}_M  &=
\sem{\Sp(\rho_{I}(\overline{x}))}_{M[\overline{x} \leftarrow \overline{u}]} 
\cup \bigcup\limits_{i =1}^{n} \sem{\Sp(t_i)}_M
\\
& \text{for $I \in \Ii$ with definition }\\
& \text{$I(\overline{x}) :=_{\mathit{lfp}} \rho_I(\overline{x})$, and}\\
& \text{where $\overline{u} = (\sem{t_1}, \dots, \sem{t_n})$}\\
\end{align*}
\end{minipage}%
\begin{minipage}{.5\linewidth}
\begin{align*}%{lcll}
\sem{\Sp(R(t_1\dots, t_n))}_M &= \bigcup\limits_{i =1}^{n} \sem{\Sp(t_i)}_M \mbox{ for $R\in \Rr$}\\
\sem{\Sp(f(t_1\dots, t_n))}_M &=
\bigcup\limits_{i =1}^{n} \{\sem{t_i}_M\} \cup \bigcup\limits_{i =1}^{n} \sem{\Sp(t_i)}_M \text{ if } f\in \Ff_m\\
\sem{\Sp(f(t_1\dots, t_n))}_M &= \bigcup\limits_{i =1}^{n} \sem{\Sp(t_i)}_M \text{ if } f\not\in \Ff_m\\
\sem{\Sp(\ite(\gamma : t_1, t_2))}_M &= \sem{\Sp(\gamma)}_M \cup
\begin{cases}
    \sem{\Sp(t_1)}_M \text{ if $M \models \gamma$}\\
    \sem{\Sp(t_2)}_M \text{ if $M \not\models \gamma$}
\end{cases}\\
\sem{\Sp(\ite(\gamma : \alpha, \beta))}_M &= \sem{\Sp(\gamma)}_M \cup
\begin{cases}
    \sem{\Sp(\alpha)}_M \text{ if $M \models \gamma$}\\
    \sem{\Sp(\beta)}_M \text{ if $M \not\models \gamma$}
\end{cases}\\
\sem{\Sp(\exists y\!: y = f(x).\;\varphi)}_M &= \{\sem{x}_M\} \cup \sem{\Sp(\varphi)}_{M[y \leftarrow u]}\\
&\text{where } u = \sem{f(x)}_M\\
\end{align*}
\end{minipage}
\vspace{-0.5cm}
\caption{Semantics of Support operator. $\sem{e}_M$ refers to the interpretation of an expression $e$ in a model $M$. The support is defined as the least interpretation satisfying the given equations.
}
\vspace{-0.2 cm}
\label{fig:sp-semantics}
\end{figure}

\subsection{Program Verification}
\label{sec:triples}

We now describe the programming language and the notion of correctness for which we develop automation in this work. These concepts serve as background for the verification condition (VC) generation presented in Section~\ref{sec:vcgen}, and the reader may safely skip it until after Section~\ref{sec:slfl}.

Figure~\ref{fig:proglang} describes the syntax of the language. It supports the typical commands including mutation of fields, allocation, and deallocation of locations. The language also supports function calls, including recursive calls~\footnote{We do not include \textit{while} loops in the formal syntax to simplify the technical exposition. However, our theory readily extends to programs with iteration, with specifications including frame logic assertions as loop invariants.}. The operational semantics for programs in this language is the usual one for heap programs ( described in Appendix~\ref{app:op-sem}). Program configurations are of the form $(S, H, A)$ where $S$ is a store, $H$ is a heap, and $A$ is the set of allocated locations. There is also an \emph{error} configuration $\bot$. We denote a transition between configurations $C_1$ and $C_2$ on a program $P$ according to the operational semantics by $C_1 \xrightarrow{P} C_2$.

\begin{figure}
\raggedright
    $P \coloneq  x\coloneq y \mid 
        x\coloneq c \mid
        v \coloneq be \mid  
        x\coloneq y.f \mid
        v \coloneq y.d \mid
        y.f\coloneq x \mid
        y.f\coloneq c \mid
        y.d\coloneq v \mid
        y.d\coloneq be $\\
        \hspace*{0.54cm}
        $\mid
        \alloc(x) \mid 
        \free(x) \mid 
        \bar{q} \coloneq g(\bar{p}) \mid 
        \assume(\eta) \mid 
        P;P \mid$
        if $\eta$ then $P$ else $P \mid
        \mathit{return} $\\
    \vspace{-0.2cm}
    \caption{The syntax of the programming language. Here, $x,y$ are location variables of type $Loc$, $c$ is a location constant, $f$ is a pointer of type $Loc \rightarrow Loc$, $d$ is a data field of type $Loc \rightarrow \tau_{\text{bs}}$ for some background sort $\tau_\text{bs}$, $be$ is a background expression, $v$ is a variable of a background sort, $\eta$ is a boolean expression without any dereferences. Finally, $g$ is a function of type $\tau_{1}\times\dots\times\tau_{m} \rightarrow \tau^{\prime}_{1}\times\dots\times\tau^{\prime}_{n}$ for some $m,n$ where $\tau_{i}, \tau^{\prime}_{i} \in \Ss$. %F?
 This is a method whose body is itself a program whose variables are a superset of the input variables $p_{i}$ of type $\tau_{i}$ and the output variables $q_{j}$ of type $\tau^{\prime}_{j}$. }
    \label{fig:proglang}
\vspace{-0.5 cm}
\end{figure}
%\medskip

\mypara{Hoare Triples and Validity} We consider triples of the form $\{\alpha\}\, P\, \{\beta\}$ %where the precondition $\alpha$ and the postcondition $\beta$ 
where $\alpha$ and $\beta$ are FL formulas in the fragment described above and $P$ is a program. We treat free variables in $\alpha$ and $\beta$ as constants, implicitly quantifying over them universally. We then define:

%We define such a triple to be \emph{valid} if

\begin{definition}[Hoare Triple Validity]
\label{defn:triple-validity}
$\{\alpha\}\, P\, \{\beta\}$ is \emph{valid} if for every %valid
configuration $(S,H,A)$ such that $(S,H) \models \alpha$ and $\sem{\Sp(\alpha)}_{(S,H)} = A$:
\begin{enumerate}
    \item $(S,H,A)$ does not transition to $\bot$ on $P$ according to the operational semantics, and
    \item if $(S,H,A) \xrightarrow{P}(S^{\prime}, H^{\prime}, A^{\prime})$, then $(S^{\prime}, H^{\prime}) \models \beta$ and $\sem{\Sp(\beta)}_{(S^{\prime},H^{\prime})} = A'$ 
\end{enumerate}

Note that we require the allocated set in the post-state to be \emph{equal} to the support of the postcondition.
\end{definition}

Informally, the above definition says the Hoare Triple is valid if starting from any configuration satisfying the precondition where the allocated set is precisely the support of the precondition, (a) the program does not behave erroneously (e.g., make unsafe dereferences), and (b) if it reaches a final configuration then the postcondition must hold, with the allocated set in the post state precisely equal to the support of the postcondition. To ensure the former, at each step in the program we must check that dereferenced elements are contained \emph{within} the allocated set. This motivates the following definition, which we will use during VC generation for \FL\ in Section~\ref{sec:vcgen}.

\smallskip
\noindent
\textit{Relaxed Postconditions.} We introduce a variant $RP$ on postconditions, read \emph{relaxed post}, to indicate that the allocated set in the post-state need not be `tight' for the postcondition. We denote these triples by $\{\alpha\}\,P\,\{RP\!: \beta\}$. Their correctness is defined similarly to that of %The meaning of this is the same as 
$\{\alpha\}\, P\, \{\beta\}$ above, except in condition (2) 
we only require $\sem{\Sp(\beta)}_{(S^{\prime},H^{\prime})} \subseteq A'$. Similarly, we also use the variant $HP$ to denote \emph{Heapless Postconditions}, with 
 no checks on the allocated set. 

%% file: seplogic.tex
%\newpage
\section{A Separation Logic with Frame Logic Inspired Semantics}
\label{sec:slfl}

We now define a separation logic \SLFL\  with \emph{alternate semantics} guided by frame logic semantics, and an FO-complete procedure for it.

\emph{An entire illustrated example of a program with \SLFL\ specifications, its translation to \FL\ and the verification condition generation for it is given in  Appendix~\ref{app:eval}.}

The primary design principle is that formulas have \emph{unique} heaplets under which they can hold. 
In particular, these heaplets are inspired by the support operator in FL. There are several salient features of our separation logic: (1) Given a global heap $g$ and a store $s$, any formula $\alpha$ will hold in at most one heaplet of $g$, which we call $\Supp(\alpha, s, g)$, (2) Heaplets are assumed to be arbitrary universes of locations (finite or infinite), avoiding incompleteness that finite models bring, and (3) We disallow recursive definitions that take parameters from background sorts. Recursive definitions are solely defined over the location sort, though the definitions themselves can utilize background sorts by dereferencing pointers on locations. This allows us to eventually translate formulae to the logic $\nplogic$ which has FO-complete reasoning. 
For example, rather than have a recursive definition $mem(x,i)$ that returns membership of an integer in a list pointed to by $x$, in our logic we would define $keys(x)$ that collects the keys stored in the list pointed to by $x$, and assert $i \in keys(x)$. Note that such definitions, translated to first-order logic, avoid quantification over the background sorts, which is crucial for completeness.
Finally, we do not generate verification conditions in separation logic itself, but rather translate separation logic annotations to frame logic, and utilize the VC generation for frame logic which we develop in the previous section. Generating precise verification conditions in separation logic often calls for magic-wand operations that we avoid by working with frame logic.

Our logic is powerful, subsuming several known precise separation logics in the literature~\cite{framelogictoplas2023,BerdineCalcagnoOHearn2005,ohearn04} (precise separation logics are also designed to have unique heaplets). In particular, our logic allows disjunction for spatial formulae which are not supported by prior precise separation logics. We also allow for recursive definitions that evaluate to background sorts that are constrained by arbitrary first-order theories such as arithmetic. 

While our alternate semantics adhere broadly to the design principles of separation logics, the semantics of disjunction turns out to be starkly different, in order to ensure unique heaplets. However, our logics provides an \textit{if-then-else} ($\ite$) construct, which has a more traditional semantics, and that can model splitting into cases.
Disjunction is often used for case analysis (e.g., definitions of linked lists that use disjunction to capture whether $x$ is $nil$ or non-$nil$, is really a case-split). 
From our experience in annotating programs in our evaluation, it turns out that disjunction can mostly be avoided (entirely avoided in our benchmarks), and replaced with such case analysis. 

We define first a base separation logic \SLFLb, where the semantic differences are clear, and a 
base separation logic \SLFLbd, with disjunction added. We then extend it to a more expressive logic (\SLFL) with inductive definitions and background sorts. All these logics can be translated to FL and hence, using our program verification scheme for frame logic,  we provide FO-complete program verification for them.

\subsection{Base logics \SLFLb\ and \SLFLbd}

Let us fix a set of locations $Loc$.
Let $Loc?$ denote $Loc \cup \{\textit{nil}\}$ where $\textit{nil}$ is a special symbol and $\textit{nil} \not \in Loc$.
Let us fix a countable set of variables $\textit{Var}$, and let $x, y, x_1, x_2, x', y'$ etc. range over $\textit{Var}$.
These variables will be used to intuitively model program variables as well as quantified variables.
Let us fix a finite set of pointers $\textit{Ptr}$, and let $f, f', g$, etc. range over $\textit{Ptr}$.

The syntax for the logic \SLFLb\ is given in Figure~\ref{fig:base-slfl-syntax}. Here, $x, y$ range over $\textit{Var}$ and $f \in \textit{Ptr}$. 
The logic supports the separating conjunction $\star$, conjunction $\wedge$, a \emph{weak conjunction} $\weakconj$,
a special if-then-else construct $\textit{ite}$, and guarded existential quantification. The logic
\SLFLbd\ is identical but includes the disjunction operator.

\mypara{Semantics:}
A \emph{store} is a function $s: \textit{Var} \rightarrow Loc?$. A \emph{heaplet} $h$ is a set of functions $\{h_f \mid f \in \textit{Ptr}\}$, where each $h_f: H \rightarrow Loc?$, where $H \subseteq Loc$ is the common domain for all the functions $h_f$. Note that we do not assume heaplets are finite.\footnote{In our logics, heaps can be infinite, and so can heaplets, like the heaplet of an infinite list.} Note that the domain of all the functions $h_f$ is common. We denote by $h(f)$ the function $h_f$ that $h$ defines, and by 
$dom(h)$ the set of locations $H$.
Also, for any $S \subseteq dom(h)$, let $h \proj S$ denote the heaplet where the domains of the pointer fields in $h$ are restricted to $S$. 
A heaplet $h'$ is a subheap/subheaplet of $h$ if there is some $S \subseteq dom(h)$ such that $h'=h \proj S$.
A global heap is a heaplet $h$ with $dom(h)=Loc$.

\begin{figure}
\[
\begin{array}{rl}
\text{\SLFLb~ formulas~} \alpha,\beta, \gamma& \coloneq
\delta \mid x \carrow{f} y 
\mid \alpha \star \beta 
\mid \alpha \wedge \beta 
\mid \alpha \weakconj \beta 
%madhu: removing disjunction from logic
%\mid \ \alpha \vee \beta 
%\\
%& & & 
\mid ~~\textit{ite}(\gamma, \alpha, \beta)
~~\mid~~ \exists y. (x \carrow{f} y:   \alpha) \\
\text{\SLFLbd~ formulas~}  \alpha,\beta, \gamma& \coloneq  
\delta \mid x \carrow{f} y 
\mid \alpha \star \beta 
\mid \alpha \wedge \beta 
\mid \alpha \weakconj \beta
\mid \alpha \vee \beta
%madhu: removing disjunction from logic
%\mid \ \alpha \vee \beta 
%\\
%& & & 
\mid ~~\textit{ite}(\gamma, \alpha, \beta)
~~\mid~~ \exists y. (x \carrow{f} y:   \alpha) \\
\text{H.I.  formulas~}  \delta & \coloneq  \textit{true}  
\mid \textit{false}  
\mid x=y 
\mid x\not = y
\mid x=\textit{nil}
\mid x\not =\textit{nil}
\end{array}
\]
\vspace{-0.4 cm}
\caption{Base Separation Logics with FL inspired semantics \SLFLb\ and \SLFLbd; H.I.:``heap-independent''}
\label{fig:base-slfl-syntax}
\vspace{-0.4 cm}
\end{figure}

\begin{figure}\small
\vspace{-0.4 cm}
\[
\begin{array}{rl}
\Supp(\delta, s, h) & =  \emptyset, \textit{~for~any~heap-independent~atomic~formula~} \delta \\
\Supp(x \carrow{f} y, s, h) & =  
\{s(x)\} \textit{~if~} s(x) \textit{~is~in~} dom(h)
\textit{~and~} \bot \textit{~otherwise}\\
\Supp(\alpha \oplus \beta, s, h) & = \Supp(\alpha, s, h) \cup \Supp(\beta, s, h), \textit{~where~} \oplus \in 
\{ \wedge, \star, \weakconj\}\\
%\Supp(\alpha \star \beta, s, h) & = & \Supp(\alpha, s, h) \cup \Supp(\beta, s, h) & \\
%\Supp(\alpha \vee \beta, s, h) & = & \Supp(\alpha, s, h) \cup \Supp(\beta, s, h) & \\
\Supp(\textit{ite}(\gamma, \alpha, \beta), s, h) 
%& = \bot \textit{~if~} \Supp(\gamma, s, h)=\bot, \textit{~or~} (\Supp(\alpha, s, h)=\bot \textit{~and~} \Supp(\beta, s, h)=\bot)\\
& =  
\Supp(\gamma, s, h) \cup
\Supp(\alpha, s, h)
\textit{,~if~} (s, h \!\proj\! \Supp(\gamma, s, h) \models \gamma \\
& \hspace*{4cm}\textit{~~~~~~~~~~~~~~~~~and~}(s,h\!\proj\! \Supp(\alpha, s, h)) \models \alpha \\
& = \Supp(\gamma, s, h) \cup\Supp(\beta, s, h), \textit{otherwise} \\
%\textit{~and~}
%\Supp(\beta, s, h) \textit{~otherwise} \\
\Supp(\exists y. (x \carrow{f} y: \alpha),s,h)  & =    
  \bot \textit{~~if~} s(x) \not \in dom(h)\\
& = 
  \{s(x)\} \cup \Supp(\alpha, s[y\mapsto h(f)(s(x))], h) \textit{~otherwise}.
\end{array}
\]

\[
\begin{array}{rclcl}
(s,h) & \models & \delta & \textit{iff} &  dom(h)=\emptyset 
\textit{~and~} s \models \delta, \textit{~for~any~heap-independent~formula~} \delta \\

(s, h) & \models & x \carrow{f} y & \textit{iff} & dom(h)=\{s(x)\} \textit{~and~} h(f)(s(x))=y\\

(s, h) & \models & \alpha \star \beta & \textit{iff} &  \textit{there~exists~} h_1, h_2 \textit{~subheaplets~of~} h, \dom(h_1) = \Supp(\alpha, s, h), \\
&&&& ~ \dom(h_2) = \Supp(\beta, s, h), 
\dom(h_1) \cup \dom(h_2)=\dom(h), \\
&&&& ~~
\dom(h_1) \cap \dom(h_2)= \emptyset,  (s,h_1) \models \alpha \textit{~and~} (s,h_2) \models \beta\\

(s, h) & \models & \alpha \wedge \beta & \textit{iff} &  (s,h) \models \alpha \textit{~and~} (s,h)\models \beta\\

(s, h) & \models & \alpha \weakconj \beta & \textit{iff} &  \textit{there~exists~} h_1, h_2 \textit{~subheaplets~of~} h,
\dom(h_1) \cup \dom(h_2)=\dom(h), \\
&&&& ~~
(s,h_1) \models \alpha \textit{~and~} (s,h_2) \models \beta\\

(s, h) & \models & \textit{ite}(\gamma, \alpha, \beta) & \textit{iff} &  
\textit{there~exists~} h_1, h_2 \textit{~subheaplets~of~} h,
\dom(h_1) \cup \dom(h_2)=\dom(h),\\
&&&& ~~
\dom(h_1)=\Supp(\gamma, s, h_1), \textit{~and,~} \\
&&&& ~~
[~ (s, h_1)     \models \gamma \textit{~and~} \dom(h_2)=\Supp(\alpha, s, h_2) 
\textit{~and~} (s,h_2) \models \alpha)~] or\\
&&&& ~~ 
[~(s,h_1) \not \models \gamma \textit{~and~} 
\dom(h_2)=\Supp(\beta, s, h_2) 
\textit{~and~} (s,h_2) \models \beta)~]\\

(s, h) & \models & (\exists y. x \carrow{f} y: \alpha)
 & \textit{iff} &  
 x \in dom(h) \textit{~and~}
(s[y \mapsto h(f)(x)], h) \models \alpha
\end{array}
\]
~\\
\vspace{-0.1 cm}
{\bf Rules for disjunction:}
\[
\begin{array}{rl}
\Supp(\alpha \vee \beta, s, h) & = \Supp(\alpha, s, h) \cup \Supp(\beta, s, h)
\end{array}
\]
\[
\begin{array}{rclcl}
(s, h) & \models & \alpha \vee \beta & \textit{iff} &  \textit{there~exists~} h_1, h_2 \textit{~subheaplets~of~} h,
\dom(h_1) \cup \dom(h_2)=\dom(h),\\
&&&& ~~
\dom(h_1)=\Supp(\alpha, s, h_1),
\dom(h_2)=\Supp(\beta, s, h_2), 
and,\\
&&&& ~~ ( (s,h_1) \models \alpha \textit{~or~} (s,h_2) \models \beta)
\end{array}
\]
\vspace{-0.4 cm}
\caption{Definition of supports and semantics for base separation logic \SLFLb, mutually defined. Definitions for \SLFLbd\ include the rules for disjunction. Set operations on supports that evaluate to $\bot$ evaluate to $\bot$.}
\label{fig:slflbsemantics}
\vspace{-0.5 cm}
\end{figure}

 The semantics of the base logics are defined in Figure~\ref{fig:slflbsemantics}. Formulas $\alpha$ are interpreted on a pair $(s, h)$, where $s$ is a store and $h$ is a heaplet, as in traditional separation logic. However, the semantics are mutually recursively defined along with a definition of a \emph{support operator} $\Supp$, where $\Supp$ maps each formula to a subset of locations, or to $\bot$ that stands for it being undefined. Intuitively, the support of a formula is the subset of locations of $h$ that the evaluation of a formula in the heaplet $h$ depends upon, and is $\bot$ if $\alpha$ cannot be evaluated within the heaplet $h$.

As we shall prove soon (Lemma~\ref{lemma52}), the definitions are designed so that whenever $(s, h) \models \alpha$ holds, it will be the case that $dom(h) =\Supp(\alpha, s, h)$, i.e., $\alpha$ can hold only on heaplets where the support map evaluates to the entire heaplet. Assuming this, let us go through the semantics. 
 
The semantics for heap-independent formulas requires heaplets to be empty. 
This allows us to ensure that the heaplet is determined, rather than in traditional separation logic where such formulas hold over any heaplet. Note that we don't have an $\textit{emp}$ formula, as $\textit{true}$ has the same semantics.
The semantics of
$x \carrow{f} y$, $\alpha \star \beta$, and $\alpha \wedge \beta$ are similar to standard separation logic semantics~\cite{reynolds02}. Though the definition for $\star$ says any way to split the two heaplets into disjoint parts that satisfy the two subformulas is fine, we know inductively that there is in fact only one way to effect this split.
The semantics of conjunction and weak conjunction are similar (the weak conjunct doesn't require the two heaplets to be the same). 

The $\ite$ construct has an involved semantics. Intuitively, to evaluate $\ite(\gamma, \alpha, \beta)$, we would evaluate $\gamma$ \emph{on the heaplet it requires}, and depending on whether $\gamma$ holds or not, evaluate $\alpha$ on the heaplet \emph{it requires}. Note that the heaplet required by any formula $\alpha$ is the subheaplet $h'$ such that the support of $\alpha$ computed with respect to $h'$ is precisely $dom(h')$. 
The definition of support is carefully designed so that such a subheaplet, if it exists, is \emph{unique} (proven below in Lemma~\ref{lemma51}). 

The definition of supports of formulas with respect to a store and heaplet $\Supp(\alpha, s, h)$, are natural definitions, but note that the definition of $\ite$ refers to the semantics of heaplets in order to split cases on whether $\gamma$ is satisfied by a subheaplet or not. 

Intuitively, while $(s,h) \models \alpha$ captures whether a formula holds on a heaplet or not, $\Supp(\alpha, s, h)$ being $\bot$ or not captures whether the heaplet is large enough to evaluate $\alpha$, and is crucial for the semantics of the \emph{else} branch of $\ite$ expressions.
 
\paragraph{Disjunction:} 
The semantics of disjunctive formulas $\alpha \vee \beta$ is quite different
from traditional semantics. Traditional semantics would declare this to hold if either $\alpha$ or $\beta$ holds. But in a larger heaplet $h$, there could be both a heaplet where $\alpha$ holds and one where $\beta$ holds, and hence this goes against our desire to have unique heaplets for formulae. In our semantics, the support of such $\alpha \vee \beta$ is considered to be the \emph{union} of the supports of $\alpha$ and $\beta$, and 
$\alpha \vee \beta$ holds in this combined heaplet. 

The logic \SLFLbd\ includes disjunction, but as we mentioned earlier, we did not see a true need for disjunction in our experimental evaluation as most uses are expressible using $\ite$ expressions.

\paragraph{Properties:} 
The supports of formulas have several properties worth noting. 
Let $S$ be the support of a formula with respect to a store and a heaplet $h$, with $S \not = \bot$. 
First, $S$ will be 
a subset of $dom(h)$. Second, consider a heaplet $h'$ that agrees with $h$ on $S$. Then
the support of the formula with respect to $h'$ will be $S$ as well. Third, the support of the formula with respect to $h$ restricted to $S$ will be $S$ itself. Finally, there is at most one subheaplet $h'$ of $h$ such that the support of the formula in $h'$ is $dom(h')$ itself. The following lemma formalizes this (see Appendix~\ref{app:baseseplogic} for a proof gist).

\begin{lemma}\label{lemma51}
    Let $s$ be a store and $\alpha$ be a \SLFLbd\ formula, $h$ be a heaplet, and let $S=\Supp(\alpha, s, h)$, and $S \not = \bot$.
    \begin{enumerate}
        \item $S \subseteq dom(h)$.
        \item Let $h'$ be a heaplet such that $S \subseteq dom(h')$ and 
        $h' \proj S = h \proj S$. Then  $\Supp(\alpha, s, h')=S$.
         \item $\Supp(\alpha, s, h \proj S)= S$.
        \item Let $h_1$ and $h_2$ be two sub-heaplets of $h$ and assume $\Supp(\alpha,s,h_1)=dom(h_1)$ and $\Supp(\alpha,s,h_2)=dom(h_2)$. Then $dom(h_1)=dom(h_2)$. \qed
    \end{enumerate}
\end{lemma}

The above semantics ensures a crucial property--- consider any store $s$, heap $h$, and formula $\alpha$, then there is \emph{at most} one sub-heaplet $h'$ of $h$ that satisfies $\alpha$ (namely $h\proj \Supp(\alpha, s, h)$). This property does not hold for standard separation logic semantics (see Appendix~\ref{app:baseseplogic} for a proof gist).

\begin{lemma}\label{lemma52}
 Let $s$ be a store and $h$ a heaplet, and $\alpha$ be an \SLFLbd\ formula.
\begin{enumerate}    
    \item  If $(s, h) \models \alpha$, then $\Supp(\alpha,s,h) = \dom(h)$.
    \item  There is at most one subheaplet $h'$ of $h$ such that $(s,h') \models \alpha$. \qed
\end{enumerate}
\end{lemma}

\mypara{Translation to Frame Logic}

\begin{figure}\small
%\[
\noindent\begin{minipage}{.5\linewidth}
\begin{align*}
\Pi(\gamma) &= \gamma, \textit{~for~H.I.~atomic~formula~} \gamma \\
\Pi(\alpha \star \beta) &= \Pi(\alpha) \wedge  \Pi(\beta) \wedge \Sp(\Pi(\alpha)) \cap \Sp(\Pi(\beta)) = \emptyset\\
\Pi(\alpha \wedge \beta) &= \Pi(\alpha) \wedge  \Pi(\beta) \wedge \Sp(\Pi(\alpha)) =  ~~\Sp(\Pi(\beta))\\
\Pi(\alpha \weakconj \beta) &= \Pi(\alpha) \wedge  \Pi(\beta)
\end{align*}
\end{minipage}%
\begin{minipage}{.5\linewidth}
\begin{align*}
\Pi(x \carrow{f} y) &= f(x)=y  \\ 
\Pi(\alpha \vee \beta) &= \Pi(\alpha) \vee  \Pi(\beta)\\
\Pi( \textit{ite}(\gamma, \alpha, \beta) ) &=  \textit{ite}(\gamma, \Pi(\alpha), \Pi(\beta)) \\
\Pi(\exists y. (x \carrow{f} y: \alpha)) &= \exists y\!: y = f(x).\, \Pi(\alpha)
\end{align*}
\end{minipage}%\]
\vspace{-0.2cm}
\caption{Translation from \SLFLbd to Frame Logic}
\label{fig:sltofl}
\vspace{-0.5 cm}
\end{figure}

\SLFLbd formulas, with their semantics using determined heaplets/supports, can be readily translated to frame logic. The translation is given in Figure~\ref{fig:sltofl}, and is simple and natural.
The translation preserves both truthhood as well as heaplet semantics in the following sense: 
\begin{lemma}\label{lemma53}
    Let $g$ be a global heap (a heaplet with domain $Loc$) and $s$ be a store. For frame logic formulas, we interpret the store as an FO interpretation of variables.
    Let $\alpha$ be an \SLFLbd\ formula.
    Then 
    \begin{enumerate}
        \item $g \models_s \Pi(\alpha)$ iff there exists a heaplet $h$ of $g$ such that
        $(s,h) \models \alpha$.
        \item If $g \models_s \Pi(\alpha)$, then
         $\Supp(\alpha, s, g)$ is equal to the value of $\Sp(\Pi(\alpha))$ in $g$. \qed
    \end{enumerate}
\end{lemma}

The above allows us to build automatic verification procedures for programs annotated with \SLFLb and \SLFLbd\ formulas. Annotations with contracts using separation logic formulas can be translated to Frame Logic, with the understanding that they define implicit heaplets which is the support of the corresponding Frame Logic formula (see Appendix~\ref{app:eval} for an example).

\subsection{Extending the base logic to background sorts and recursive definitions}
We now extend the base logic to a more powerful logic that (a) allows recursive definitions of predicates and functions, and (b) incorporates background sorts (like arithmetic), pointer fields to such sorts,  and recursively defined functions that evaluate to the background sort.

\smallskip
Let us fix background sorts, with accompanying functions and relations, $\mathcal{G}=\{g_1, \ldots, \}$ and $\mathcal{P} = \{p_1, \ldots, \}$. Let us also fix a set of symbols for functions and relations, $\mathcal{F}$ and
$\mathcal{R}$, respectively, which we will use to define new \emph{recursively-defined} functions and relations.
Heaplets are extended to include pointers that map locations to background sorts.

\begin{figure}
\vspace{-0.5cm}
\[
\begin{array}{rrcl}
\text{SL-FL Formulas} & \alpha,\alpha',\beta & \coloneq & 
\textit{true}  
\mid \textit{false}  
\mid x=y 
\mid x\not = y
\mid x=\textit{nil}
\mid x\not =\textit{nil}
\mid x \carrow{f} y \\
& & & 
\!\!\!\!\!\!\!\!\!\!\!\!\!\!\!\!\!\!\!\!\!\mid p(\overline{t})
\mid \alpha \star \beta 
\mid \alpha \wedge \beta 
\mid \alpha \weakconj \beta
\mid \ \alpha \vee \beta 
\mid \textit{ite}(\alpha', \alpha, \beta)
\mid \exists y. (x \carrow{f} y: \alpha) 
\mid R(\overline{x}) \\

\text{Terms} & t, t', t_i & \coloneq & 
x \mid t.f \mid g(\overline{t}) \mid \textit{ite}(\alpha, t, t') \mid F(\overline{t})\\
\textit{Rec. Defs:} & R(\overline{x}) &  :=_\textit{lfp} & \rho_R(\overline{x}, \mathcal{R}, \mathcal{F}), \textit{~for~each~} R \in \mathcal{R} \\
        &  F(\overline{x}) &  :=_\textit{lfp} & \mu_F(\overline{x}, \mathcal{R}, \mathcal{F}), \textit{~for~each~} F \in \mathcal{F} 
\end{array}
\]
\vspace{-0.5cm}
\caption{\SLFL: Syntax of full Separation Logic with Frame-Logic inspired semantics. The guards in $\textit{ite}$ expressions should not mention recursively defined relations $R \in \mathcal{R}$}
\label{fig:slfl}
\vspace{-0.5 cm}
\end{figure}

The syntax of \SLFL\ is given in Figure~\ref{fig:slfl}. The syntax for \SLFL\ formulas is similar to the base logic, but we allow them to evaluate predicates over terms of background sorts ($p(\overline{t})$), and allow evaluation of recursively defined predicates ($R(\overline{x})$). 
The syntax for terms allows dereferencing locations (resulting in either locations or background sort),
computing functions over background sorts (like $+$ over integers), if-then-else constructs, or recursively defined functions that return values of background sort (like "Keys" of a location $x$ pointing to a list). 

Recursively defined predicates ($R$) and functions ($F$) are defined with parameters of type location \emph{only}. This is a crucial  restriction as we will, upon translation to Frame Logic and later to first-order logic, treat these definitions as universally quantified equations, and it's important for FO-completeness using natural proofs that these quantifications are over the foreground sort of locations only. Definitions of predicates ($R$) and functions ($F$) are mutually recursive and are 
arbitrary \SLFL\ formulae $\rho$ and terms over background sorts $\mu$ that can mention $\mathcal{R}$ and $\mathcal{F}$.

The semantics of the logic extend that of the base logic in the natural way, and we skip the formal semantics here (see Appendix~\ref{app:seplogic} for some details).
We assume that each background sort is a complete lattice (for flat sorts like arithmetic, we can introduce a bottom element and a top element to obtain a complete lattice). The support maps are extended for recursive definitions, and they evaluate to the support of their definitions. The semantics of both $\Supp$ and $\models$ are taken together as equations and their semantics is defined as the least fixpoint of these equations.

The separation logic that emerges is powerful and can state properties such as standard datastructures (lists, trees) and properties of them (keys, bst), etc. 
Furthermore, the logic can be translated easily to FL (extending the $\Pi$ translation above) for automated verification of programs annotated using \SLFL\ (see Appendix~\ref{app:seplogic} for some details).

%% file: vcgen.tex
\section{Automating Program Verification for Frame Logic}
\label{sec:vcgen}

In this section, we present the second technical contribution of this paper: a VC generation technique for programs annotated in \FL\ . To this extent, we extend \FL\ with the `cloud' operator, which enables quantifier elimination. We then utilize strongest postconditions for VC generation. This avoids the nightmarish formulas present in the weakest precondition computations in~\cite{esop2020framelogic, framelogictoplas2023}. 

\input{cloud-operator}

\medskip
\noindent

We now turn to VC generation for programs annotated with quantifier-free FL specifications, now extended with the cloud operator. Note that the user continues to write specifications in FL with guarded quantification as in Figure~\ref{fig:syntax}. These are eliminated automatically using the transformation described in Lemma~\ref{lem:cloud-lemma}.

We fix some notation. For a given program, we assume a fixed set of pointer and data fields $\Ff_m$, and we use $f, f', f_1$, etc. to denote these. We also fix a finite set of recursive definitions $\Ii$, and symbols $I, I', I_1$, etc. range over $\Ii$. We denote by $\textit{Def}$ the recursive definitions for each $I \in \Ii$, with the definition of $I$ of the form $I(\vec{x}) =_\mathit{lfp} \rho_I(\vec{x})$. Definitions can be mutually recursive. 

Given a triple $\{\textit{Pre}\}\, P\, \{\textit{Post}\}$, we fix a set of methods $G$ which are called by $P$. We assume every $g\in G$ has its own contract with precondition $\textit{Pre}_{g}$ and postcondition $\textit{Post}_{g}$. Our verification is modular and therefore simply assumes the validity of triples for called methods. 

An illustrative example is given below. The program \egprog{} prepends the reverse of list pointed to by $x$ to the list pointed to by $y$. The precondition asserts the two lists are disjoint. The postcondition asserts the program returns a list $\ret$ whose keys is the union of the keys of the two lists at the start of the program. We use $\Old(R(x))$ to denote references to elements of the appropriate background sort obtained by calculating $R(x)$ on the program state at the start of the method call. In particular, the support of the expression $\Old(R(x))$ is empty.

\begin{footnotesize}
\begin{tabular}{l l}
\makecell[l]{
    {\bf Definitions and Specification:}\\
    $\lst(x) := \ite(x = \nil, \top, \lst(\nxt(x)) \land x \not\in \Sp(\lst([\nxt(x)]))$
    \\
    $\Keys(x) := \ite(x = \nil, \phi, \Keys(\nxt(x)) \cup \{\key(x)\})$
    \\ 
    $\text{Pre: }\lst(x) \wedge \lst(y) \wedge \Sp(\lst(x)) \cap \Sp(\lst(y)) = \phi$
    \\
    $\text{Post: } \lst(\ret) \wedge \Keys(\ret) = \Old(\Keys(x)) \cup \Old(\Keys(y))$
}
&  
\makecell[l]{
    {\bf Program:}\\
    $\text{\egprog{}}(x, y) \text{  returns}(\ret)$:
    \\
   (if (x = nil) then ret := y; \\ \hspace{0.3cm}
     else tmp := x.next; x.next := y; \\ \hspace{0.3cm}
     ret := \egprog{}(tmp, x);  )\\
    return;
    }
\end{tabular}
\end{footnotesize}

\subsection{Generating Triples over Basic Blocks}
Given a program with annotations, we first generate several triples over basic blocks. 
These generated triples capture both the verification of the postcondition for $P$ along different control flows as well as other ``side-conditions'' such as memory safety of dereferences and deallocations.

The basic blocks corresponding to a program $P$, denoted $BB(P)$, are as follows: 
\begin{itemize}
\item $BB(c) = \{c\}$ for any atomic command $c$ 

\item $BB( \textit{if } \eta \textit{ then } P_1 \textit{ else } P_2 ) = \{ assume (\eta); P_1'\mid P_1' \in BB(P_1)\} \cup \{ assume (\neg \eta); P_2' \mid P_2' \in BB(P_2)\}$

\item $BB( P_1; P_2) = \{ P_1'; P_2' \mid P_1' \in BB(P_1), P_2' \in BB(P_2) \}$
\end{itemize}

$BB(P)$ computes basic blocks in the usual manner, splitting paths on conditionals and including the condition in conditionals using $\mathsf{assume}$ statements. Each basic block is assumed to be in ``SSA form'', i.e., every variable is assigned at most once. This is not a restriction as every basic block in our language can be transformed into SSA form.  

Given a triple $\{\textit{Pre}\}\,P\,\{\textit{Post}\}$, the set of triples associated with the basic blocks $BB(P)$ are:
\begin{itemize}
 \item For each $s \in BB(P)$,$~~~~~~~~~~~$
    $\{\textit{Pre}\}\, s\, \{\textit{Post}\}$

 \item For every $s_1, s_2$, where 
    $s_1; x:=y.f; s_2$ or $s_1; y.f := x; s_2$ or $s_1; \textit{free}(y); s_2$ is in $BB(P)$,

    \begin{center}
    $\{\textit{Pre}\}\, s_1\, \{HP: y \in A\}
    \footnote{ Later in this section, we will show that the allocated set $A$ can be represented as an FL formula. Thus, checks such as $y \in A$ can be captured in FL.}
    $    
    \end{center}
    
    where $HP$ is Heapless Post (see Section~\ref{sec:triples}).

 \item For every $s_1, s_2$ such that 
    $s_1; \bar{b}:= g(\bar{a}); s_2 \in BB(P)$,

    \begin{center}
    $\{\textit{Pre}\}\, s_1\, \{RP: \textit{Pre}_g[\bar{p} \gets \bar{a}]\}$    
    \end{center}
    
    where $RP$ is Relaxed Post (see Section~\ref{sec:triples}).

\end{itemize}

The first item above includes the verification of the postcondition of the program along every basic block. 
The second set of Hoare triples captures safety of dereferences and ensures freed locations are currently allocated. For each statement that dereferences a location variable or frees a location, we add a Hoare triple for the prefix of the basic block till that statement and check using relaxed postcondition that
the variable is in the current allocated set. 
Finally, for each statement that calls a method $g$, we introduce a prefix of the basic block till the call and check, again using a relaxed postcondition, that the precondition of the called method holds.
The relaxed postcondition above ensures that the heaplet which is the support of the precondition of $g$, $Pre_g$ (which $g$ will remain within) is a subset of the currently allocated set. 

One can think of generating the triples for BBs as stepping through the program one command at a time and discharging necessary obligations. 
Let us consider the 'else' branch of the illustrative example. Assuming $x\not=\nil$, we first check if dereferencing $x$ is safe when assigning $x.\nxt$ to $\tmp$ with the top-left Hoare triple in Figure~\ref{fig:example-BBs}. The case for the mutation $x.\nxt$ to $\tmp$ is similar. Next, we require the precondition for $\egprog{}$ to hold on the method call, i.e., with inputs $\tmp, x$. Moreover, we also ensure its support is contained in the allocated set. We treat the support of the precondition to the call as the set of modified locations for the function call. This is captured by the top-right Hoare triple in Figure~\ref{fig:example-BBs}. Finally, we check that the post-condition holds at the end of the program by the bottom triple in Figure~\ref{fig:example-BBs}.

\begin{figure}[H]
\footnotesize
\centering
\begin{tabular}{l l}
    \makecell[c]{
    $\{ \lst(x) \wedge \lst(y) \wedge \Sp(\lst(x)) \cap \Sp(\lst(y)) = \phi \}$\\
    $\assume(x \not= \nil); \tmp := x.\nxt$\\
    $\{ HP: x \in A\}$
    }
     &  
    \makecell[c]{
    $\{ \lst(x) \wedge \lst(y) \wedge \Sp(\lst(x)) \cap \Sp(\lst(y)) = \phi \}$\\
    $\assume(x \not= \nil); \tmp := x.\nxt; x.\nxt := y $\\
    $\{ RP: \lst(\tmp) \wedge \lst(x) \wedge \Sp(\lst(\tmp)) \cap \Sp(\lst(x)) = \phi \}$
    }
\end{tabular}
\phantom{~}\\
    $\{ \lst(x) \wedge \lst(y) \wedge \Sp(\lst(x)) \cap \Sp(\lst(y)) = \phi \}$\\
    $\assume(x \not= \nil); \tmp := x.\nxt; x.\nxt := y;$\\
    $ \ret := \text{\egprog{}}(\tmp,x)$\\
    $\{ \lst(\ret) \wedge \Keys(\ret) = \Old(\Keys(x)) \cup \Old(\Keys(y)) \}$
\caption{Hoare Triples generated for the `else' branch in the illustrative example.
}
\label{fig:example-BBs}
\end{figure}

\vspace{-0.5 cm}
\subsection{Verification Condition Generation}

We present the core technical contribution in this section, a verification condition (VC) generation mechanism for basic blocks. We generate verification conditions in frame logic with only universal quantification over locations (for satisfiability) and with functions mapping from the foreground sort to the foreground or background sorts (integers, sets of locations, etc.).  

Intuitively, we will be constructing, for each Hoare triple $\{\textit{Pre}\}\, s\, \{\textit{Post}\}$ a VC of the form $((\textit{Pre} \wedge T) \Rightarrow \textit{Post}) \wedge SC$ where $T$ captures the semantics of transformation the basic block $s$ has on the state and heap, and $SC$ is a support condition that demands that the support of $\{Post\}$ is the support of $\textit{Pre}$ modulo allocations and freeing of locations that happen in $s$ (including calls to other functions). In the case of Hoare triples with relaxed postconditions $\{\textit{Pre}\}\, s\, \{RP: \textit{Post}\}$, the VC is similar, except that the support condition will demand that the support of the postcondition is a subset of the support of the precondition, modulo allocations and freeing of locations. Heapless postconditions $\{\textit{Pre}\}\, s\, \{HP: \textit{Post}\}$ simply do not have a support condition.

Consider a Hoare triple $\{\alpha\}\, s\, \{\beta\}$ where $s$ is a basic block and where 
$\alpha, \beta$ are annotations in frame logic.
We use the notion of a \emph{local configuration} to build the VC. We describe a transformation of local configurations across statements of a basic block, processing one statement of the basic block at a time, left to right. 
A local configuration is a 5-tuple $(T, A, H, \textit{Fr}, \textit{RD})$. After processing a prefix $s'$ of a basic block $s$,  $T$ corresponds to the formula describing the state after the transformation $s'$ has had on states satisfying the precondition. The component $A$ represents the set of allocated locations after the prefix, 
$H$ represents the heap after executing the prefix, and $\textit{Fr}$ stands for a set of frame rules gathered during mutations of the heap by the prefix. Finally, $\textit{RD}$ represents a set of new recursive definitions gathered to represent the recursive definitions $\Ii$ evaluated on intermediate heaps during the execution of the prefix $s'$. 

Initially, $T$ is set to  $\textit{Pre}$ and $A$ is set to the support of the precondition $\alpha$ of the Hoare triple, i.e., $\Sp(\alpha)$, $H$ is initially simply a set of uninterpreted functions denoting the initial heap, and $\textit{Fr}$ is the empty set. The set $\textit{RD}$ is initialized to $\Dd$, the recursive definitions of $\Ii$ on the initial heap. Formally:

\begin{definition}
A \emph{logical configuration} is a tuple $(T, A, H, \textit{Fr}, \textit{RD})$ where:
\begin{itemize}
 \item $T$ is a quantifier-free logical formula describing the program transformation. $T$ is expressed over program variables and recursive functions defined in $\textit{RD}$
 \item $A$ is a term whose type is a set of locations, and denotes the current allocated set of locations.
 \item $H$ is a map that captures the current heap; it associates with each pointer $f \in F$ a map
 $H(f): Loc \xrightarrow{} \sigma^{\prime}$, where $\sigma'$ is the foreground sort of Locations or a background sort. $H(f)$ is expressed as $\lambda u. t(u)$, where $t$ is a term over $u$ and the program variables.  

\item $\textit{Fr}$ is a set of universally quantified FO formulae that denotes a set of \emph{frame rules} gathered for each mutation of the heap.

\item $\textit{RD}$ is a set of recursive definitions of a set of function symbols that capture properties/data structures of intermediate heaps.
\end{itemize}
\end{definition}

The heap map $H$ associates with each pointer $f \in F$, a map $H(f)$, that maps locations to the appropriate sort (either location sort or a background sort). Note that these hence satisfy the one-way condition we want for FO-completeness.
The map $H(f)$ is itself expressed in logic using quantifier-free terms involving other function symbols, such as the functions that characterize the initial heaplet. For example, $H(f)$ may be the formula $\lambda u. ite(u=z, g_1(z), g_2(u))$, where $g_1, g_2$ are function symbols and $z$ is a program variable. In general, we will expand the signature with new function symbols as we process a basic block. 

We start with a set of recursive definitions for function symbols in  $\Ii$. 
As we construct the VC,
we adapt these definitions of functions in $\Ii$ to ones that refer to them interpreted on new heaps. 

In particular, for any $H$ and $I \in \Ii$, let us introduce a new function symbol
$I_H$, and let us add a recursive definition for
$I_H$. This recursive definition $I[H]$ 
is the definition for $I$ where every pointer and data field $f \in \Ff_m$ mentioned in the recursive definition is replaced by $H(f)$. 
In other words, for each recursively defined function $I \in \Ii$, where $I$ is defined as 
$I(x)=_\textit{lfp} \rho(x)$, we give a definition $I[H]: I_H(x) = _\textit{lfp} \rho[H(f)/f](x)$, where $\rho[H(f)/f]$ is $\rho$ where every occurrence of $f$ is replaced with the term $H(f)$, for each $f \in \Ff_m$.
%We will refer to the recursive definition of
%$R_H$ as $\Ii[H]$.

For example, consider the example of $H$ mentioned above, and consider the definition
$\textit{lseg}(x,y) =_\textit{lfp} (x=y) \vee 
  \textit{lseg}(f(x),y)$. Then 
  $\textit{lseg}[H]$ is the definition $\textit{lseg}_H =_\textit{lfp}  (x=y) \vee 
  \textit{lseg}(\textit{ite}(x=z, g_1(z), g_2(x)),y)$.

The generation of VCs are presented in Figure~\ref{fig:vc}. However, before diving into explaining the rules, let us discuss the frame rule (formally also described in Figure~\ref{fig:vc}) that is used in the rules.

\medskip
\noindent
{\bf Frame conditions:}
For any recursively defined function symbol $I(\vec{y}) \in \Ii$, two heap maps $H$, $H'$, and a set of locations $X$, we define the frame rule to be 
$$\textit{fr}(I, X, H, H') ~~:~~ \forall \vec{y} \in Loc, ( X \cap \Sp(I_H(\vec{y}))  = \emptyset) \implies (I_{H'}(\vec{y}) = I_H(\vec{y}) )$$

The above formula is meant to be used when a heap $H$ is transformed to a heap $H'$ where mutations happen only on the locations in $X$. It says that if the support of the recursive definition on $\vec{y}$ (evaluated in $H$) does not intersect $X$, its value in $H'$ is the same as that in $H$. Note that this formula has only universal quantification over the foreground sort of locations, and hence its translation to 
FORD will be in the $L_\textit{oneway}$ fragment.

We use $\textit{fr}(X,H,H')$ to denote the conjunction of  $\textit{fr}(I, X, H, H')$, for each $I \in \Ii$, which states the frame condition for each recursively defined function in $\Ii$.

\medskip
\noindent {\bf Verification Condition Generator:}
\begin{definition}[VC for a basic block]
\label{defn:bb-vc}
For Hoare triples of basic blocks $\{\alpha\}\, s\, \{\beta\}$ and for basic blocks with relaxed post, the associated VCs are defined as: 
 $$ \mathit{VC}(\{\alpha\}\, s\, \{\beta\}) =  \left(\bigwedge \textit{Fr} \wedge T\right) \Rightarrow ( \beta \wedge \textit{Sp}(\beta)=A)$$
 $$ \mathit{VC}(\{\alpha\}\, s\, \{RP: \beta\}) =  \left(\bigwedge \textit{Fr} \wedge T\right) \Rightarrow ( \beta \wedge  \textit{Sp}(\beta) \subseteq A )$$
 $$ \mathit{VC}(\{\alpha\}\, s\, \{HP: \beta\}) =  \left(\bigwedge \textit{Fr} \wedge T\right) \Rightarrow ( \beta )$$
where $(T, A, H, \textit{Fr}, \textit{RD}) = \mathit{VC}( (T_0, A_0, H_0, \textit{Fr}_0, \textit{RD}_0), s)$ where $T_0 = \alpha$, $A_0 = \textit{Sp}(\alpha)$,
$H_0(f)=\lambda u. f(u)$ (for each $f \in F$), 
$\textit{Fr}_0 = \emptyset$, $\textit{RD}_0 = \textit{RD}$, and $\mathit{VC}$ is the transformer defined formally in Figure~\ref{fig:vc}.
\end{definition}

\begin{figure}\small
\begin{mathpar}
\inferrule*[right = Assn]{T' = T \wedge (x=y)}
{\vc{\sigma}{ x \coloneq y} = (T', A, H, \Fr, \RD)}

\inferrule*[right = SclrAssn]{T' = T \wedge (v = \mathit{be})}
{\vc{\sigma}{ v \coloneq \mathit{be}} = (T', A, H, \Fr, \RD)}

\inferrule*[right = Deref]{T' = T \wedge (x = H(f)(y))}
{\vc{\sigma}{ x \coloneq y.f} = (T', A, H, \Fr, \RD)}

\inferrule*[right = Mutation]{
H'(f) = \lambda\marg.~~ite(\marg = y, ~x, ~H(f)(\marg)) \text{ for every $f' \in \Ff_m, f' \not = f$}
\\
\Fr' = \Fr \cup \{ \fr(\{y\}, H, H') \}
\\
\RD' = \RD \cup \{ I[H'] \mid I \in \Ii\}
}
{\vc{\sigma}{ y.f \coloneq x} = (T, A, H', \Fr', \RD')}

\inferrule*[right = Alloc]{
T' = T \wedge (x\not\in A)\\
A' = A \cup \{x\}\\
H'(f) = \lambda \marg.~~\ite( \marg = x, ~\defaultf, ~H(f)(\marg)), \text{ for every $f\in \Ff_m$, with domain $A'$}
\\
\RD' = \RD \cup \{ I[H'] \mid I \in \Ii\}
}
{\vc{\sigma}{\alloc(x)} = (T', A', H', \Fr, \RD')}

\inferrule*[right = Free]{
A' = A \backslash \{x\}\\
H'(f) = H(f) \text{ for every $f\in \Ff_m$, with domain $A'$}\\
\RD' = \RD \cup \{ I[H'] \mid I \in \Ii\}
}
{\vc{\sigma}{\free(x)} = (T, A', H', \Fr, \RD')}

\inferrule*[right = SeqComp]{\vc{\sigma}{P} = \sigma'}
{\vc{\sigma}{P;Q} = \vc{\sigma'}{Q}
}

\inferrule*[right = Assume]{T' = T \wedge \alpha}
{\vc{\sigma}{\assume(\alpha)} = (T', A, H, \Fr, \RD)
}

\inferrule*[right = Call]{
g \in G \\ \alpha(\bar{x}) = \mathit{Pre}(g) \\ \beta(\bar{x}, \bar{y}) = \mathit{Post}(g)
\\
H'(f) = \lambda \marg.~~\ite( \marg \in \Sp(\alpha_H[\bar{x}\gets \bar{a}]),~f_{\mathit{fresh}}(\marg),~H(f)(\marg))  
% \text{ with domain $A'$}
\\
T' = T \wedge \beta_{H'}[\bar{x}\gets \bar{a}, \bar{y}\gets \bar{b}]
\\
A' = (A \setminus \Sp(\alpha_H[\bar{x}\gets \bar{a}])) \cup \Sp(\beta_{H'}[\bar{x} \gets \bar{a}, \bar{y} \gets \bar{b}])
\\
\Fr' = \Fr \cup \{ ~\fr(\Sp(\alpha_H[\bar{x} \gets \bar{a}]), H, H') ~\}
\\
\RD' = \RD \cup \{ I[H'] \mid I \in \Ii\}
}
{\vc{\sigma}{\bar{b} \coloneq g(\bar{a})} = (T', A', H', \Fr', \RD')}
\end{mathpar}

\vspace{-0.3 cm}
\caption{Verification Condition Generation: Predicate transformers for basic blocks. Here $\sigma = (T, A, H, \Fr, \RD)$, the configuration before applying the $\mathit{VC}$ transformer, and $\fr(X, H, H') = \bigwedge_{I \in \Ii} \fr(I, X, H, H')$. We also assume a (fixed) set of callable functions $G$. Each function $g \in G$ is annotated with precondition $\mathit{Pre}(g) = \alpha(\bar{x})$ and postcondition $\mathit{Post}(g) = \beta(\bar{x}, \bar{y})$. By $\alpha_{H}$, we mean $\alpha$ where each $f \in \Ff_m$ is replaced by $H(f)$ and each $I \in \Ii$ by $I_H$. By $f_{\mathit{fresh}}$, we mean a fresh function symbol with the signature of $f$.  
\label{fig:vc}
}
\end{figure}

We now give the intuition behind the transformer $\mathit{VC}$. For an assignment $\texttt{x} := \texttt{y}$ ({\bf Assn}), we just
add the condition $x=y$ as a conjunct to our transformation formula. Recall that we have assumed that the basic block is in SSA form, and this is the sole assignment to the variable $x$, and hence this conjunct captures the program configuration constraining $x$ and $y$ to be equal. The scalar assignment rule ({\bf SclrAssn}) for assigning variables of background sorts is similar. 

For a dereference $\texttt{x}:= \texttt{y.f}$ ({\bf Deref}), the transformation looks up the current pointer $f$, using the function $H(f)$, and adds the conjunct that equates $x$ to $H(f)(y)$. 
Note that $H(f)(y)$ is the formula $H(f)$ with $y$ substituted as its parameter, and hence results in a quantifier-free term of the appropriate type. The {\bf Assn} and {\bf Deref} rules only change the transformation formula, and the rules for assume ({\bf Assume}) and sequential composition ({\bf SeqComp}) are as expected.

For a {\bf Mutation} $\texttt{y.f} := \texttt{x}$, the transformation component doesn't change, but the heap is updated to reflect the mutation, where $H'(f)(y)=x$ and is $H(f)$ on other locations. 
We also introduce the recursive definitions for $\Ii$ adjusted for the current heap $H'$ (adding them to $\textit{RD}'$). And introduce frame rules $\textit{fr}(\{y\}, H, H')$ that set these new recursive definitions on any tuple to the value of old ones provided that $y$ does not belong to the support of these definitions on the tuple. 

For a function call ({\bf Call}), we introduce an elegant way to take care of the heap transformation. For each pointer and data field $f \in \Ff_m$, we introduce a \emph{fresh} function symbol $f_{\mathit{fresh}}$. The map $H'(f)$ is now an $ite$ (if-then-else) term, where $H'(f)(u)$ evaluates to $H(f)(u)$, which is the old value of the pointer $f$, provided $u$ does not belong to the support of the precondition of the called function $g$. Otherwise, $H(f)(u)$ evaluates to $f_{\mathit{fresh}}(u)$, which is an arbitrary value since $f_{\mathit{fresh}}$ is uninterpreted. The transformation component imbibes the postcondition guaranteed by the call to $g$, where the recursively defined functions $R$ in the postcondition are modified to the functions $R_{H'}$. Recursive definitions are adapted to the new heap $H'$ after the call and added to $\textit{RD}'$. The new allocated set is derived from the old allocated set by removing the locations in the support of the precondition of $g$ and adding back the locations in the support of the postcondition of $g$. This handles the locations $g$ may allocate or free during its execution. Finally, we add frame condition formulae  
$\textit{fr}(Sp(\alpha[\bar{x} \gets \bar{a}]), M, M^{\prime})$ that says that for any recursive definition that has an empty intersection with the support of the precondition of $g$ remains unaffected in the new heap. 

The transformation for $\texttt{alloc}$ statement ({\bf Alloc}) adds an assumption that the allocated location is not already present in the current allocated set $A$, adds it to the new allocated set, and updates the heap map so that the new location's pointers point to default values. The transformation for $\texttt{free}$ statements ({\bf Free}) simply removes the location $x$ from the allocated set. Even though logical expression for the heap function $H$ does not change with freeing of locations, \emph{its domain does}. Hence, when either allocation or free happens, we recompute the recursive definitions on the new heap by updating the $\textit{RD}$ component.
Note that we do not check whether freed locations are indeed allocated as those are checked on other basic blocks, as described in the previous subsection. 

The algorithm applied to the bottom triple of Figure~\ref{fig:example-BBs} is given below. The text in black is the triple to be proven. The text in green shows the changes to the logical configuration. The first two statements simply add the appropriate term to the program transformation. For the mutation $x.\nxt := y$, we update the heap $H$ to $H_1$ reflecting this change. We then define the recursive functions on the new heap. Finally, we add frame rules relating recursive functions over $H$ to those over $H_1$. The function call $\ret := \text{\egprog{}}(\tmp,x)$ employs a similar idea. First, we must update the heap. We know the portion of the heap modified by the call is the support of the precondition for the call, i.e $\Sp(\alpha_{H_1}(\tmp, x))$.  Thus, the new heap $H_2$ will only differ outside of $H_1$. Furthermore, the postcondition $\beta_{H_2}(\tmp, x, \ret)$ for the call holds on the new heap. The new allocated set is the portion not modified by the call union the support of the postcondition the call. Finally, we add a frame rule connecting the heaps $H_1$ and $H_2$. The VC generated for this BB is $(\fr({x}, H, H_1) \wedge \fr(\Sp(\alpha_{H_1}(\tmp, x)), H_1, H_2) \wedge T_2) \implies \varphi \wedge \Sp(\varphi) = A_2$, where $\varphi = \lst_2(\ret) \wedge (\Keys_2(\ret) = (\Old(\Keys(x)) \cup \Old(\Keys(y))))$.

\begin{footnotesize}
\flushleft
$\{ \lst(x) \wedge \lst(y) \wedge ( \Sp(\lst(x)) \cap \Sp(\lst(y))  = \phi)\}$
\begin{tabular}{l l}
    $\assume(x \not= \nil);$ &
    \graytext{ \# $T \leftarrow T \wedge (x \not=\nil)$}
    \\
    $ \tmp := x.\nxt; $ &
    \graytext{ \# $T \leftarrow T \wedge (\tmp = \nxt(x))$}
    \\
    $x.\nxt := y;$ &
    \graytext{ \# $H_1 = H[\nxt \leftarrow \lambda\marg.~~\ite(\marg = x, y, \nxt(\marg))]$
    }
    \\ 
    &
    \graytext{ \# $\lst_1 = \lst[H_1],~\Keys_1 = \Keys[H_1],~\RD_1 = \{ \lst, \Keys, \lst_1, \Keys_1 \}
    $
    }
    \\
    &
    \graytext{ \# $\Fr_1 = \{\fr({x}, H, H_1)$ \}}
    \\
    $ \ret := \text{\egprog{}}(\tmp,x)$
    & \graytext{ \# $ \alpha_{H_1}(\tmp, x) = (\lst_1(\tmp) \wedge \lst_1(x) \wedge \Sp(\lst_1(\tmp)) \cap \Sp(\lst_1(x)) = \phi) $}
    \\
    & \graytext{ \# $H_2[f] = \lambda\marg.\ite(\marg\not\in\Sp(\alpha_{H_1}(\tmp, x)), H_1(f)(\marg) , \nxt_2(\marg) )$ for $f \in \{ \nxt, \key \}.$}
    \\
    &
    \graytext{ \# $ \lst_2 = \lst[H_2],~\Keys_1 = \Keys[H_2] ,~\RD_2 = \RD_1 \cup \{ \lst_2, \Keys_2 \}$}
    \\
    & \graytext{ \# $\beta_{H_2}(\tmp,x, \ret) = (\lst_2(\ret) \wedge \Keys_2(\ret) = \Old(\Keys_1(\tmp)) \cup \Old(\Keys_1(x)))$  
    }
    \\
    & \graytext{ \# $T_2 = T \wedge \beta_{H_2}(\tmp,x, \ret)$
    }
    \\
    &
    \graytext{\# $A_2 = (A\backslash \Sp(\alpha_{H_1}(\tmp, x))) \cup \Sp(\beta_{H_2}(\tmp,x, \ret))$}
    \\
    & \graytext{ \# $\Fr_2 = \Fr_1 \cup \{ \fr(\Sp(\alpha_{H_1}(\tmp, x)), H_1, H_2) \}$
    }
\end{tabular}
$\{ \lst_2(\ret) \wedge \Keys_2(\ret) = (\Old(\Keys(x)) \cup \Old(\Keys(y))\}$

\end{footnotesize}

\paragraph{Soundness and Completeness} Our VC generation is sound and complete for the validity of Hoare Triples over basic blocks. Fix a triple $\{\alpha\}\, s\, \{\beta\}$ where $s$ is a basic block, and let $\mathit{VC}$ be the map defined in Definition~\ref{defn:bb-vc}. Then:

\begin{theorem}[Soundness of VC Generation for Basic Blocks]
\label{thm:vcgen-soundness}
If $\mathit{VC}(\{\alpha\}\, s\, \{\beta\})$ is a valid formula in Frame Logic then $\{\alpha\}\, s\, \{\beta\}$ is valid in the sense of Definition~\ref{defn:triple-validity}. Similarly, if $\mathit{VC}(\{\alpha\}\, s\, \{RP\!: \beta\})$ is valid in Frame Logic then the triple $\{\alpha\}\, s\, \{RP\!: \beta\})$ is valid. 
\end{theorem}

\begin{theorem}[Completeness of VC Generation for Call-Free Basic Blocks]
\label{thm:vcgen-completeness}
Let $s$ be a basic block with no function calls. If $\{\alpha\}\, s\, \{\beta\}$ is valid (Definition~\ref{defn:triple-validity}) then $\mathit{VC}(\{\alpha\}\, s\, \{\beta\})$ is valid. Similarly, if $\{\alpha\}\, s\, \{RP\!: \beta\})$ is valid then $\mathit{VC}(\{\alpha\}\, s\, \{RP\!: \beta\})$ is valid.
\end{theorem}

We elide the proofs of these theorems as they are fairly straightforward from the definition of the VC transformer and the operational semantics (Appendix~\ref{app:op-sem}). The key argument in the proof of completeness is that we model mutations \emph{precisely} using updates to the heap map $H$. Similarly, allocation and deallocation are also modeled precisely using the symbolic allocated set $A$, which corresponds to the true allocated set as defined by the operational semantics for call-free blocks. 
Though the above theorem does not include function calls, our VC is also precise for basic blocks with function calls, with respect to a \emph{modular operational semantics}, where function calls are assumed to havoc the precondition's heaplet and replace it with one that satisfies the postcondition. 

Our result is interesting in its own right, and showcases the power of our framework. Contemporary works on VC generation for other logics such as Separation Logic are often not complete~\cite{BerdineCalcagnoOHearn2005}.

\subsection{Validating Verification Conditions}
\label{sec:vcreasoningshort}

In previous sections we described how we reduce validity of Hoare triples annotated with quantifier-free FL formulas to the validity of FL formulas. We then reason with the generated VCs in FL by (a) translating them to equi-valid FORD formulas, and (b) using \emph{natural proofs}, an FO-complete automation technique developed in prior work~\cite{qiu13,pek14,loding18} for checking validity of FORD formulas in the $\nplogic$ fragment. Natural proofs is based on systematic quantifier instantiation and SMT solving. The translation from FL to FORD in the first stage follows a construction similar to the one developed in earlier works on FL~\cite{esop2020framelogic,framelogictoplas2023}, but we take some additional care to extend the translation to handle the cloud operator and to ensure that the translations of our VCs fall into $\nplogic$. Intuitively, the translation encodes the semantics of the $\Sp(\cdot)$ and $\Cl{\cdot}$ operators in FORD itself. We describe both stages in detail in Appendix~\ref{sec:vcreasoning}.

%% file: cloud-operator.tex
\subsection{Simplifying Quantification Using the Cloud Operator}
\label{sec:cloud}

In this work we wish to write quantifier-free specifications, generating quantifier-free verification conditions in the $\nplogic$ fragment. However, defining something even as simple as linked lists in FL requires quantification~\cite{framelogictoplas2023,esop2020framelogic}:
\begin{center}
$\lst(x) :=_{\mathit{lfp}} \ite(x=\nil, \top, \exists y : y = \nxt(x).\, \lst(y) \land x \notin \Sp(\lst(y) ))$    
\end{center}

\noindent
where an existential quantifier is used in the above definition to say that when $x \neq \nil$, then $y$, the ``next'' location of $x$ must point to a linked list such that $x$ does not belong to the support of $\lst(y)$\footnote{For readers familiar with separation logic, this formula is similar to the SL formula $\exists y.\, (x \overset{\nxt}{\mapsto} y) * \lst(y)$}. 

It is tempting to think that the quantifier can be eliminated by simply replacing $y$ with $\nxt(x)$. However, the second conjunct then becomes $x \notin \Sp(\lst(\nxt(x)))$, which does not hold because the support of a formula that mentions $\nxt(x)$ always contains $x$\footnote{The example shows that substitution is not a semantics-preserving operation in FL. Indeed, a substitution is only valid if the replacement expression has the same truth value \emph{and} the same support as the original substituted sub-expression.}.

In this work we introduce a new operator called the \emph{Cloud} operator, denoted by square brackets $\Cl{\cdot}$. The cloud of a formula (resp. term) $\alpha$ is a formula $\Cl{\alpha}$ that evaluates the same way as $\alpha$, but whose support is empty. Simply, $\Cl{\alpha}$ treats $\alpha$ as a support-less expression. We extend FL with the cloud operator, which allows us to rewrite the above definition without quantification:

\begin{center}
$\lst(x) :=_{\mathit{lfp}} \ite(x=\nil,\, \top,\, \nxt(x) = \nxt(x) \land \lst(\Cl{\nxt(x)}) \land x \notin \Sp(\lst(\Cl{\nxt(x)})))$
\end{center}

\noindent
where all occurrences of $y$ are replaced with $\Cl{\nxt(x)}$. We also add the tautological conjunct $\nxt(x) = \nxt(x)$ to ensure that the support of the resulting formula still contains $\{x\}$.

\noindent
Formally, we define the semantics of the cloud operator as follows:

\begin{center}
$\Sp(\Cl{\varphi}) = \emptyset \textrm{ and } M \models \Cl{\alpha} \textrm{ iff } M \models \alpha  \textrm{ for a formula $\alpha$}$\\
$\Sp(\Cl{t}) = \emptyset \textrm{ and } \sem{~\Cl{t}~}_M = \sem{t}_M \textrm{ for a term $t$}$
\end{center}

We then have the following lemma. Fix a model $M$.
\begin{lemma}[Eliminating Quantification using the Cloud Operator]
\label{lem:cloud-lemma}
$M \models \exists y\!: y = f(x).\; \varphi$ if and only if $M \models (f(x) = f(x)) \land \varphi(\Cl{f(x)})$, and $\sem{\Sp(\exists y\!: y = f(x).\; \varphi)}_M = \sem{\Sp((f(x) = f(x)) \land \varphi(\Cl{f(x)}))}_M$\qed
\end{lemma}

We skip the proof of this lemma as it follows trivially from definitions.

%% file: evaluation.tex
\section{Implementation and Evaluation}
\label{sec:eval}

\subsection{Implementation}
We develop the Frame Logic Verifier (\textsc{FLV}) tool based on the mechanism described in Sections~\ref{sec:vcgen} and~\ref{sec:vcreasoningshort}. The tool allows users to write heap-manipulating programs and annotations in quantifier-free Frame Logic, using the support $\Sp(\cdot)$ and cloud $\Cl{\cdot}$ operators. The tool also allows users to write inductive lemmas which are proven automatically using a form of induction using pre-fixpoints~\cite{loding18}.

The FLV tool generates VCs that are translated to FORD. These quantifier-free FORD formulas are then validated using an existing solver for natural proofs~\cite{fossil} which implements the procedure described in Section~\ref{sec:natproofs}, instantiating recursive definitions and reasoning with the resulting quantifier-free formulas using {\sc Z3}~\cite{Z3}.   

The implementation is also optimized in several ways. First, we introduce new recursive definitions and frame rules only at certain key points--- before and after function calls, and at the end of a basic block. Note that the formulation in Section~\ref{sec:vcgen} introduces them across every mutation, allocation, deallocation, and function call (as described in Section~\ref{sec:vcgen}). These optimized frame rules `batch' the individual mutations, catering to the \emph{footprint} of the program between the key points. Second, the tool starts off with a thrifty set of instantiations commonly needed, based on previous work on natural proofs(~\cite{pek14}): recursive definitions and lemmas are instantiated on the relevant footprint from the last key point, and frame rules are instantiated on the current variables and one level of pointer dereferences on them before and after function calls and at the end of the basic block. Finally, we use an SMT solver to simplify the translations to FORD before validating them.

We develop the \SLFL\ verifier (SLFLV) on top of FLV. This tool takes programs annotated with \SLFL\ sepcfications. SLFLV translates these to programs with \textit{FL} annotations (see Figure~\ref{fig:sltofl}) and the resulting program is verified by FLV. 

The translation from \SLFL\ to \textit{FL} incurs some translation bloat in formula size. 
We develop an optimized version of the tool that controls this bloat, especially in nested conjunctions.
First, while translating \SLFL\ formulas to FL formulas bottom-up, we compute, for every \SLFL\ subformula $\alpha$, a simple formula FL formula $\alpha^\dagger$ such that $\Sp(\Pi(\alpha))=\Sp(\alpha^\dagger)$, and whenever we need to use the support of $\Pi(\alpha)$, we use the support of $\alpha^\dagger$ instead. In particular, for $\alpha=\beta \wedge \gamma$, we take
$\alpha^\dagger$ to be simply $\beta^\dagger$ (as the translation independently constrains $\Sp(\beta^\dagger)= \Sp(\gamma^\dagger)$). Second, we syntactically track subformulas whose support is empty, and 
remove constraints that check whether the intersection with these supports is empty. We evaluate the effect of these optimizations in our evaluation.

\subsection{Benchmark Suite}

We evaluate the expressiveness and performance of our contributions on a suite of data structure programs. Each program is annotated separately with \textit{FL} and \SLFL\ specifications. Our suite consists of 29 programs involving operations on data structures such as singly and doubly linked lists, circular lists, binary search trees, red-black trees, treaps, etc. We obtained this suite by distilling a core set of benchmarks from prior work~\cite{pek14}. We considered the various crucial data structure manipulation algorithms for each structure such as insertion, deletion, search, sorting, and traversals. A summary of our benchmark suite can be found in 
Appendix~\ref{app:eval}, along with examples of specifications, lemmas used, and an illustrative example of the VC generation.

\subsubsection{Specifications} We annotate our programs with complete functional specifications.
The contracts specify precisely the data structure returned and its content in terms of the effect the routine has on the input data structures and parameters.
We define data structures and various measures over them using recursive definitions and write specifications using them. To aid in writing post-conditions, we introduce an $\mathit{Old}$ operator. For a recursive function $R$, $\mathit{Old}(R(x))$ is equivalent to assuming $r_x = R(x)$ at the beginning of a program for a fresh variable $r_x$ and using that in place of $\mathit{Old}(R(x))$. In particular, this means that $\Sp(Old(R(x)))$ is empty as it is simply syntactic sugar for a variable. 

\subsubsection{Lemmas}
\label{sec:lemmas-discussion}

Though our procedures are FO-complete, i.e guaranteed to validate VCs when they are valid with respect to fixpoint interpretations of definitons, there is a gap between this semantics and the least fixpoint semantics of recursive definitions~\cite{loding18}. This gap is 
bridged using user given \emph{inductive lemmas}. Inductive lemmas are expected to be proved using a simple induction where the lemma itself acts as the induction hypothesis, and using a first-order proof without least fixpoint interpretations (see~\cite{qiu13,pek14,loding18}). 

Lemmas are typically required to relate different properties of a data structure, and are typically independent of the methods that work on the data structures (e.g., a sorted list is a singly linked list).

Our tools prove these lemmas using induction and natural proofs, and utilize the proven lemmas while proving verification conditions. 
Lemmas for benchmarks are currently written in frame logic. However, we can  facilitate writing lemmas in \SLFL\ as well, where we interpret an entailment $\alpha \models \beta$ as $\Pi(\alpha) \implies \Pi(\beta) \wedge Sp(\Pi(\alpha)) = Sp(\Pi(\beta))$.

Note that there has been a lot of orthogonal work that has addressed the problem of automatically synthesizing these lemmas, both for functional and imperative programs, with considerable success (incomplete, of course). See~\cite{ranjitlemmasynthesis24, fossil, fluid23, reynolds15, yang19, sivaraman22}.

\subsubsection{Definitions that have the same supports}

%EqSp
The user is also allowed to declare that the supports of a pair of recursive definitions is equal. For example, over trees, the support of $\mathit{BST}$, $\mathit{Keys}$, $\mathit{Min}$, $\mathit{Max}$, $\mathit{Height}$ are all the same.
The tool proves this property (treating this as an inductive lemma) and cuts down on the number of recursive definitions of supports, which in turn leads to better performance.

\begin{table}[t]\footnotesize
    \centering
    \begin{tabular}[t]{l l | r r r r r}\hline
Data
&Benchmark & Number of  &Number of  & \textit{FL} Verif. &  \SLFL\  Verif. Time & \SLFL\ Verif. Time \\

Structure
 &&Basic Blocks & Permission Checks  & Time& (unoptimized) & (optimized)   \\ \hline
\multirow{10}{*}{SLL}
& Append        &$2$ & $5$ & $1$s & $1$s & $1$s \\
& Copy All      &$2$ & $6$ & $1$s & $1$s & $1$s\\
& Delete        &$3$ & $9$ & $1$s & $1$s & $1$s\\
& Find          &$3$ & $5$ & $1$s & $1$s & $1$s\\
& Insert Back   &$2$ & $7$ & $1$s & $1$s & $1$s\\
& Insert Front  &$1$ & $3$ & $1$s & $1$s & $1$s\\
& Reverse       &$3$ &  $7$ &$1$s & $1$s & $1$s\\
&Insertion Sort    &$5$  & $14$ &  $3$s & $4$s  & $3$s \\
&Merge Sort        &$10$ & $26$ & $6$s & $7$s & $6$s \\
&Quick Sort        &$9$  & $25$ &  $13$s & $18$s& $15$s \\

\hline

\multirow{4}{*}{DLL}
& Insert Back   &$2$ & $9$ & $1$s & $1$s & $1$s \\
& Insert Front  &$2$ & $8$ & $1$s & $1$s & $1$s \\
& Delete Mid    &$4$ & $11$& $2$s & $3$s & $3$s \\
& Insert Mid    &$4$ & $10$& $2$s & $3$s & $3$s \\

\hline

\multirow{3}{*}{\makecell[l]{Sorted \\ List}}
& Delete     &$3$ &  $9$ &  $1$s & $1$s & $1$s\\
& Find       &$4$ & $6$ &  $1$s  & $1$s & $1$s\\
& Insert     &$3$ & $10$ & $2$s  & $3$s & $3$s\\

\hline

\multirow{3}{*}{\makecell[l]{CL}}
& Delete Front   &$5$ & $12$ & $6$s & $6$s & $6$s \\
& Find           &$6$ & $10$ & $2$s  & $3$s  & $3$s \\
& Insert Front   &$2$ & $6$  & $2$s & $2$s  & $2$s \\

\hline

\multirow{4}{*}{\makecell[l]{BST}}
& Delete      &$8$   & $27$ &  $8$s & $11$s & $10$s\\
& Find        &$4$   & $8$  &  $2$s  & $2$s  & $2$s\\
& Insert      &$4$   & $13$ &  $3$s  & $4$s  & $4$s\\
& Rotate Right &$1$   & $5$  &  $1$s  & $1$s  & $1$s\\

\hline

\multirow{2}{*}{\makecell[l]{Tree \\ Traversals}}
& BST to List        &$3$ & $11$ & $15$s & $22$s & $19$s\\
& In-Order Traversal &$2$ & $8$  & $3$s  & $7$s & $8$s\\

\hline

\multirow{2}{*}{\makecell[l]{Treap}}
& Delete      &$9$ & $32$ &  $18$s & $24$s & $20$s\\
& Find        &$4$ & $8$ &  $3$s   & $3$s & $3$s\\

\hline

\multirow{1}{*}{\makecell[l]{RBT}}
& Insert     &$11$ &$83$ &$42$s &$1$m $6$s &$56$s\\

\hline

\end{tabular}
\caption{Performance of verifications tools (SLL = Singly-Linked List, DLL = Doubly-Linked List, CL = Circular List, BST = Binary Search Tree, RBT = Red-Black Tree).}
\label{tab:eval-performance}
\vspace{-1 cm}
\end{table}

\vspace{-0.2 cm}
\subsection{Evaluation}
We evaluate our tool on our suite of benchmarks. 
The experiments were performed on a MacBook Pro (Mac14,9) equipped with an Apple M2 Pro processor, 10 CPU cores, and 16 GB of RAM.

%\subsubsection{Verifying Programs}
We present our evaluation in Table~\ref{tab:eval-performance}. 
The table gives the number of basic blocks in each program and the number of permission checks--- checking whether all dereferenced locations are contained in the support of the precondition or are newly allocated, and checking if the permissions granted to functions that are called are subsets of the method's permissions and the tightness of returned permissions (support of the postcondition). The verification times for FL and the times taken by the unoptimized and optimized versions of the tool for SL benchmarks are reported. This time includes (a) translation between logics--- \SLFL\ to FL, FL to FORD--- including optimizations, as well as the verification of (b) all VCs, including permission checks, (c) lemmas, and (d) equality of supports. 

Our tool verifies all benchmarks effectively, in less than a minute for all benchmarks. 
Over all these benchmarks, the instantiation that we identify to begin with already sufficed in proving the programs correct, and no further instantiation was needed. 

We observe that \SLFL\ benchmarks take longer than \textit{FL}, as expected.
We believe to be both due to translation bloat as well as certain superfluous checks that typical formulations in separation logic mandate. However, our optimization reduces the verification time by $4\%$ on average and comes closer to the FL verification time.

\paragraph{When Provers Fail: Bug Finding and Counterexamples} There are two parts to our verification methodology: (1) generating VCs in FL (and ultimately FORD), and (2) using an FO-complete validity technique to discharge the VCs (see Section~\ref{sec:vcreasoningshort}). FL/FORD interpret recursive definitions with the intended least fixpoint semantics, but the VC validity checking techniques we employ do not--- they are only complete with respect to a fixpoint semantics. Note that the set of models where definitions respect fixpoint semantics is larger. %than the set of models which respect least fixpoint semantics. 
For example, consider the definition $\lst(x) := (x = \nil) \lor \lst(\nxt(x))$. Under fixpoint semantics, the element $e_1$ in the model $e_1 \overset{\nxt}{\rightarrow} e_2 \overset{\nxt}{\rightarrow} \nil$ would, of course, satisfy the definition of $\lst$, but it could also be interpreted as pointing to a $\lst$ in the model $e_1 \overset{\nxt}{\rightarrow} e_2 \overset{\nxt}{\rightarrow} e_1$ representing a two-element cycle. 
The latter model would not satisfy the definition under least fixpoint semantics.  
Specifically, this means that when our FORD prover fails, it fails with respect to this larger class of models.
In practice, there is also another source of failure from insufficient instantiations. Our validity technique is based on systematic instantiation of definitions (``unrolling'') and other quantified formulas (such as frame reasoning axioms). The instantiation happens in rounds, and FO-completeness holds in the limit: if the VC is valid under fixpoint semantics, the prover will eventually succeed in some round. Each round sends a decidable quantifier-free query to an SMT solver.

Overall, there are three sources of failure for the prover: (a) the VC is FO-valid but the instantiations are insufficient, (b) the VC is valid in FO+\textit{lfp} but not in FO, and (c) the VC is invalid in FO+\textit{lfp} (i.e., program/spec has a bug). The model returned by the SMT solver in a particular round only witnesses the non-provability of the VC with the available set of instantiations and does not disambiguate between the three cases. 

When we developed our benchmark suite, we identified bugs in our specs and implementations by simply isolating the basic blocks that our solver could not prove and tracing proof obligations semi-automatically. For example, we isolated conjuncts in the postcondition that could not be proven, or added extra assertions at certain intermediate points within the basic block. This helped us identify many sources of error, among them insufficient preconditions, bugs in implementation where we missed corner cases, and the need for certain inductive lemmas. Note that inductive lemmas address the second failure mode, as discussed in Section~\ref{sec:lemmas-discussion}. We also found it particularly helpful to weaken our postconditions to Relaxed or Heapless postconditions to triage bugs pertaining to proofs regarding supports and the allocated set of locations. This allowed us, for instance, to strengthen contracts for function calls with properties about supports.

There are many approaches that can help disambiguate between VCs that are truly invalid and those that are simply not FO-valid (needing an inductive lemma). One possibility is to ask an SMT solver for a true counterexample (a model with least fixpoint interpretations, as opposed to a nonstandard model) of bounded size, say 5-6 locations. This is hard in general as one needs to express the unbounded computation involved in determining least fixpoints, but establishing practical bounds on the length of the computation may work as well. However, we believe that the design of our logics makes this approach particularly attractive, and leave it to future work to determine a general technique for synthesizing counterexamples to invalid formulas in FO-complete heap logics.

%% file: relatedwork.tex
\section{Related Work}
\label{sec:relwork}

The work on Frame Logic~\cite{framelogictoplas2023,esop2020framelogic} is the most relevant prior work, which we have discussed at length in the paper. There are many other logics~\cite{reynolds02,RegionLogic,kassios06,idf,pek14,bobot12,hobor13} for heaps that make different choices in terms of whether heaplets of formulae are implicit, whether formulas have unique heaps, etc. The most popular among these is Separation logic~\cite{demri15,seplogicprimer,ohearn01,reynolds02} which has implicit but non-unique heaplets for formulas. 
%
%introduces special symbols $*$ and $-*$ (magic wand) representing separating conjunction and separating implication, with a tight frame semantics for each statement; formulas hence implicitly define \spps{}. 
%Separation logic uses these symbols to avoid any explicit frame annotations in reasoning about programs. 
%
%For example, the associated frame rule in separation logic that a statement executed in a smaller state $P$ can be executed in any larger state satisfying $P * R$. Crucially, the special symbol $*$ here is used to ensure the frames of $P$ and $R$ are disjoint, and no explicit frame annotations are needed for reasoning about these
%
Dynamic Frames~\cite{kassios06,DynamicFrames2011} and related approaches like Region Logic~\cite{RegionLogic1,RegionLogic2,RegionLogic} allow users to explicitly specify heaplets and uses them to reason about the program's modifications.
The work on Implicit Dynamic Frames~\cite{smans12,chalice,idf,parkinson11} combines the ideas of implicit heaplets in separation logic %(without magic wand) 
and explicitly accessible heaplets in dynamic frames, resulting in implementations such as Chalice~\cite{parkinson11} and Viper~\cite{vipertool,viper-vcgen-technique,vipercav24}. In particular, Viper provides a general framework into which many permission logics can be encoded, including Separation Logic. However, it generates quantified verification conditions to SMT solvers over inherently incomplete logics. There is also work on translating VeriFast~\cite{verifast} predicates into Implicit Dynamic Frames~\cite{jost14}.

There is rich prior work on reasoning with separation logic. A first category of prior works develop fragments with decidable reasoning~\cite{Smallfoot,BerdineCalcagnoOHearn2004,BerdineCalcagnoOHearn2005,CookHaaseChristoph2011,PerezAntonioRybalchenko2011,PerezRybalchenko2013,PiskacWiesZufferey2013,guarded-wand-pagel2020,twbcade13}. There is also work on decidability of heap properties specifiable in EPR (Effectively Propositional Reasoning)~\cite{Itzhaky2014,Itzhaky2013,ItzhakyEPRInvariants}. A second category of works translate separation logic to first-order logic and use reasoning mechanisms for FOL~\cite{ChinDavidNguyen2007,pek14,PiskacWiesZufferey2013,PiskacWiesZufferey2014,PiskacWiesZufferey2014Tool,madhusudan12,qiu13,suter10,loding18}. There are also techniques that reason with separation logic using cyclic proofs~\cite{brotherston11,ta16}.

The work in~\cite{bobot12} develops a logic similar to FL and uses Why3~\cite{filliatre13esop,why3logic} to reason with verification conditions. Recent work~\cite{pldi2024IDS} identifies a %promising 
technique that avoids incompleteness and the use of inductive lemmas, but requires developing a set of maps upfront for a set of properties and have to be redesigned when the properties change.

We use in our work the reasoning technique of natural proofs~\cite{madhusudan12,qiu13,pek14,loding18}. The technique of quantifier instantiation used by natural proofs is similar to many other unfolding-based techniques in the literature~\cite{suter10,nguyen08,kaufmann97,verifast,leino12}.

The work on FL formulates a fragment of separation logic called Precise Separation Logic which has determined, unique heaplets and provides a translation of the fragment to FL. Other prior literature has also studied fragments with determined heaplets~\cite{qiu13,pek14,ohearn04} as it is easier to develop automated reasoning for them. We believe that the techniques developed in this work can open up new pathways to automating separation logics.

\section{Conclusions}
\label{sec:conclusion}
% \adithya{Notes:
% - viper-like implementation
% - lemma synthesis
% - slcomp with this semantics
% }

We have presented two logics, a frame logic and a separation logic, which admit FO-complete program verification, where the latter demands that all valid theorems under fixpoint semantics of recursive definitions are eventually proved by a system. Our work navigates various threats to incompleteness, carefully designing logics with particular semantics as well as reasoning procedures for them that 
are FO-complete, and are also effective in practice.

The most promising future direction that our work paves is to build new program logics and verification paradigms that are held to the theoretical standard of being FO-complete. Our evaluation suggests that by having FO-complete procedures that only make repeated calls to SMT solvers on decidable logics, 
efficient and FO-complete procedures are possible. Realizing such an FO-complete reasoning for a rich programming language through a mature VC generation pipeline. such as the Viper/Boogie pipeline, is hence a pressing direction to pursue.

In order to reason with first-order logics with recursive definitions, recent work~\cite{fossil,ranjitlemmasynthesis24} has proposed inductive lemma synthesis algorithms. % that work in tandem with natural proofs. 
Incorporation of lemma synthesis into our verification framework would alleviate the current burden on the programmer to write these inductive lemmas (described for our benchmarks in~Section~\ref{sec:eval}). 

We believe that our alternate semantics for separation logic \SLFL\  is worthy of further study. In particular, since we have shown that its reasoning can be automated using FO reasoning, a track in an SL competition like {\sc SLComp}~\cite{slcomp} would be interesting to encourage competitive tool development.

\newpage

%% file: dataavailability.tex
\section*{Data-Availability Statement}
Reviewers may find the \SLFL\ and \FL\ benchmarks evaluated on in the following anonymous repository: \url{https://github.com/verifproj/complete-logics-bench}. The VC-generation tool will be released for artifact evaluation.

%% file: appendix.tex
\input{app-sourcesofinc}

\input{app-opsem}

\input{app-fltranslation}

\input{app-seplogic}

\input{naturalproofsandlemmas}

\input{app-eval}

%% file: app-sourcesofinc.tex
\section{Additional Details from Section~\ref{sec:sourcesofinc}}
\label{app:sourcesofinc}

\begin{proof}[Proof of Theorem~\ref{thm:so-high-undec}]
We show a stronger result, namely that the logic is \emph{highly undecidable}. To achieve this, we reduce validity in MSO to the non-halting problem for Turing Machines whose definition of halting is that the final state must be hit infinitely often. Highly undecidable problems are not recursively enumerable, and their complements are not recursively enumerable either, i.e., they are R.E Hard $\cap$ co-R.E Hard. We describe the key elements of the proof below.

First, we define natural numbers with comparison, using symbols $\mathit{zero},\mathit{succ}$, and $<$. We constrain these symbols using quantified FO formulae in the usual way, e.g., requiring that $<$ forms a total order. Note, however, that a model may have elements beyond the $\omega$-prefix that correspond to the ``standard'' natural numbers. To cull out the $\omega$-prefix, we write the following formula quantified over a set variable $S$:

$\mathit{TrueNat} \equiv \forall S.\, \big((\mathit{zero} \in S \land (\forall x.\, x \in S \Rightarrow \mathit{succ}(x) \in S)) \Rightarrow (\forall y.\, y \in S)\big)$

\noindent
$\mathit{TrueNat}$ says that every set that contains zero and is successor-closed is universal. This can only hold on a model if every element in the model corresponds to a standard natural number.

\smallskip
The remainder of the proof is standard. We encode the Turing Machine using functions $\mathit{state}(i)$ to denote the state of the TM at time step $i$, $\mathit{head}(i)$ to denote the position of the head at time step $i$, and $\mathit{cell}(i,j)$ to denote the contents of the $j$-th cell at time step $i$. We express FO-quantified constraints to ensure that the functions are defined correctly. Let us denote these constraints by $\mathit{TMEncoding}$.

We are now ready to construct the formula to express non-halting. Recall that our definition of non-halting is that the TM does not hit the halting state infinitely often, i.e., there is a point in time after which the TM does not hit the halting state. Let the halting state be encoded $\mathit{zero}$, without loss of generality. The formula that expresses non-halting is the following one:
\begin{center}
$\mathit{TrueNat} \Rightarrow \big(\mathit{TMEncoding} \Rightarrow (\exists z\forall t.\, z < t \Rightarrow \mathit{state}(t) \neq \mathit{zero})\big)$    
\end{center}

This completes the proof. Note that the formula does not contain any alternating quantification over sets.
\end{proof}

\begin{proof}[Proof of Theorem~\ref{thm:sl-fp-incomplete}]
We first define the variant of Separation Logic that we study in this result. The syntax of the logic is as follows:
\begin{align*}
\varphi \coloneq\,& \mathsf{emp} \,|\, \mathit{true} \,|\, \mathit{StackFormula} \,|\, x \carrow{f} y \,|\, \varphi \land \varphi \,|\, \varphi * \varphi\\
&\,|\, \ite(\mathit{StackFormula}, \varphi, \varphi) \,|\, \mathit{RecPred}(\overline{t}) \,|\, \exists y.\, \varphi  \\
\mathit{StackFormula} \coloneq\,& t = t \,|\, t \neq t
\end{align*}
\noindent
where $\mathit{RecPred}$ denotes recursively defined predicates, and $t$ denotes a term. The grammar for terms is the usual one.

A recursive predicate $R(\overline{x})$ is defined using a formula $\rho(\overline{x})$ from the above logic, where $\rho$ can mention $R$ (only positively). We describe the semantics below. We evaluate formulas over a store $s$ (map from variables to values) and a heap $h$ (map from pointer symbols to functions) as Separation Logic does, but we now allow for $h$ to be infinite. In addition, we also require the model to interpret recursively defined predicate symbols, which we denote by $\Ii$. 

\begin{align*}
s,h,\Ii &\models \mathsf{emp} &&\; \textrm{if } h \textrm{ is empty}\\
s,h,\Ii &\models \mathit{true} &&\; \textrm{always}\\
s,h,\Ii &\models \mathit{StackFormula} &&\; \textrm{if } \mathit{StackFormula} \textrm{ evaluates to true according to } s\\
s,h,\Ii &\models x \carrow{f} y &&\; \textrm{if } \mathit{dom}(h) = \{x\} \textrm{ and } h(f)(s(x)) = y\\
s,h,\Ii &\models \alpha \land \beta &&\; \textrm{if } s,h,\Ii \models \alpha \textrm{ and } s,h,\Ii \models \beta\\
s,h,\Ii &\models \alpha * \beta &&\; \textrm{if there exist disjoint heaplets } h_1,h_2 \textrm{ such that }\\
&&&\;\hspace{1em} h_1 \uplus h_2 = h\textrm{ and } s,h_1,\Ii \models \alpha \textrm{ and } s,h_2,\Ii \models \beta\\
s,h,\Ii &\models \ite(\mathit{cond}, \alpha, \beta) &&\; \textrm{if } s \models \mathit{cond} \textrm{ and } s,h,\Ii \models \alpha \textrm{ or } s \not\models \mathit{cond} \textrm{ and } s,h,\Ii \models \beta\\
s,h,\Ii &\models \mathit{RecPred}(\overline{t}) &&\; \textrm{if } \Ii(\mathit{RecPred})(s(\overline{t})) = \mathit{true} \textrm{ and } s,h,\Ii \models \rho(\overline{t}) \textrm{ where } \rho \textrm{ is the body of } \mathit{RecPred}\\
s,h,\Ii &\models \exists y. \varphi &&\; \textrm{if there is some element } a \textrm{ of the appropriate sort such that }\\
&&&\;\hspace{1em} s[y \mapsto a],h,\Ii \models \varphi
\end{align*}

\smallskip
Further, we treat the above recursive semantics for $\mathit{RecPred}$ as a fixpoint rather than least fixpoint. To illustrate the consequence of using fixpoints, consider the following predicate $\lst(x)$ denoting (now fixpoint) linked lists:
\begin{center}
$\lst(x) \equiv \ite(x = \nil, \mathsf{emp}, \exists y.\, x \carrow{\nxt} y * \lst(y))$
\end{center}

Let $s$ be the store that interprets $x$ to $0$, and let $h$ be the infinite heaplet whose domain is the natural numbers, where the $\nxt$ pointer is interpreted by $h(\nxt)(i) = i+1$. This models an `infinite list'. Let $\Ii$ be the interpretation that marks every node in this infinite list as satisfying the predicate $\lst$. Clearly, $s,h,\Ii \models \lst(x)$ would not hold if $\Ii$ were the least fixpoint interpretation (no node would satisfy $\lst$). However, it can hold in the logic defined above as we admit any interpretation which is a valid fixpoint for the given recursive definition.

\medskip
We now present our construction for modeling non-halting in the above logic. We model the run of a two-counter machine using a linked list with a $\nxt$ pointer. We use integer-valued pointers $\mathit{state},c_1$, and $c_2$ to denote the state and the values of the two counters corresponding to the configuration encoded by a node in the linked list, and define a recursive predicate $\mathit{tcm\_run}(x)$ to check whether the linked list encodes a valid run of the machine. We assume without loss of generality that the initial state is $0$ and the halting state is $1$. The predicate $\mathit{tcm\_run}$ is then defined as follows:
\begin{align*}
\mathit{tcm\_run}(x) \coloneq\,& \ite\Big(x = \nil, \mathsf{emp},\\
&\hspace{1.5em}\exists y,q.\big( (x \carrow{\nxt} y
    \land x \carrow{\mathit{state}} q
    \land x \carrow{c_1} i
    \land x \carrow{c_2} j)\\
&\hspace{9em}*\\
&\hspace{4.3em}\ite(y = \nil, \mathsf{emp},\exists q'. ((y \carrow{\mathit{state}} q' \land y \carrow{c_1} i' \land y \carrow{c_2} j')*\mathit{true})\\
&\hspace{4.3em}\hspace{9em}\land \mathit{transition}(q,i,j,q',i',j')\\
&\hspace{4.3em}\hspace{9em}\land \mathit{tcm\_run}(y)\\
&\hspace{4.3em}\hspace{2em})\\
&\hspace{1.5em}\hspace{3em}\big)\\
&\hspace{1.2em}\Big)
\end{align*}

\noindent
where $\mathit{transition}$ is a macro expressing the correctness of the state transition for the two-counter machine. We define $\mathit{tcm}(x)$ to be the formula $\mathit{tcm\_run}(x) \land (x \carrow{state} 0 * \mathit{true})$, which denotes that $x$ points to a list with a valid run of the machine starting from the initial state. Similarly, we define $\mathit{halt}(x)$ to be the formula  $x \carrow{state} 1$.

\smallskip
Finally, we express the following entailment:

\begin{center}
$x \neq \nil \land (\mathit{tcm}(x) * \mathit{true}) \land (\mathit{halt}(y) * \mathit{true}) \models \mathit{tcm}(x) * \mathit{halt}(y) * \mathit{true}$
\end{center}

\medskip
We now claim that the above entailment is valid iff the two-counter machine does not halt. 

\paragraph{Machine Doesn't Halt $\Rightarrow$ Entailment Valid} Observe that in any $s,h,\Ii$ such that $s,h,\Ii \models \mathit{tcm}(x)$ holds, the $\omega$-chain of linked list elements starting at $x$ and closed under the $\nxt$ pointer must in fact encode a run from the initial configuration for the machine. Furthermore, $h$ must contain all the nodes in this $\omega$-chain. However, since we consider fixpoint definitions, it is possible that $h$ may contain more elements.

Let us assume that the machine does not halt. Now consider $s,h,\Ii$ such that the left-hand side of the entailment holds. In particular, we have that $s,h,\Ii \models \mathit{tcm}(x)$. Since the machine does not halt, the $\omega$-chain cannot contain an node with the halting state. In particular, we have that $s,h,\Ii \models \mathit{halt}(y)$, which means that $y$ cannot be part of the $\omega$-chain. Therefore, we can split $h$ into three disjoint parts: (a) $h_1$, containing the $\omega$-chain where $\mathit{tcm}(x)$ holds, (b) $h_2 = \{y\}$, where $\mathit{halt}(y)$ holds, and (c) any other elements. This shows that the right-hand side of the entailment holds, and we are done.

\paragraph{Machine Halts $\Rightarrow$ Entailment Invalid} If the machine does not halt, in any $s,h,\Ii$ such that $s,h,\Ii \models \mathit{tcm}(x)$, there must exist a node $e$ in the $\omega$-chain starting at $x$ that corresponds to the halting state. Let us now consider the entailment with $y$ being interpreted as $e$. Clearly, the left-hand side of the entailment is true. However, the right-hand side cannot be true since any $h$ such that $s,h,\Ii \models \mathit{tcm}(x)$ must necessarily contain the $\omega$-chain, and in particular must contain $e$. Therefore, the heaplet of $\mathit{halt}(y)$ cannot be separated from the heaplet of $\mathit{tcm}(x)$.

This completes the proof.
\end{proof}

Note that in the above proof we crucially use the fact that given an $s,h,\Ii \models \mathit{tcm}(x)$, if $h$ contains more elements than the $\omega$-chain, then we also have $s,h',\Ii \models \mathit{tcm}(x)$ where $h'$ is any fixpoint-closed subset of $h$ that contains the $\omega$-chain, and in particular can be the $\omega$-chain itself. The formula can be satisfied over multiple heaplets, and we use this in the first direction of the proof.

%% file: app-opsem.tex
\section{Operational Semantics}
\label{app:op-sem}

%\mypara{Operational Semantics} 
We consider \emph{configurations} of the form $(S, H, A)$ where $S$ is a store, $H$ is a heap, and $A$ denotes the set of allocated locations. Formally, $S$ is a partial map that interprets constants, variables, and non-mutable functions. $H$ is a tuple of maps--- one for each mutable function $f \in \Ff_m$--- whose domain is the universe of locations and range is the universe of the appropriate sort (locations for pointers and various background sorts for data fields). We model $\nil$ as a distinguished location, ensuring in our semantics that valid programs do not dereference $\nil$. Note that $S$ and $H$ together define a model where one can interpret FL formulas. We denote the satisfaction of an SL formula $\varphi$ on the model corresponding to a given $S,H$ by $S,H \models \varphi$. $A$ is a subset of the universe of locations that does not contain $\nil$. We also introduce an \emph{error} configuration $\bot$. 

We describe the operational semantics of various commands manipulating such configurations below in Figure~\ref{fig:opsem}. 
%in Appendix~\ref{app:op-sem}. 
The semantics is fairly standard, but we elucidate some key aspects here. First, the operational semantics checks that dereferences are memory safe, i.e., that the lookup or mutation of $x.f$ only occurs when $x$ belongs to the allocated set $A$ %(which obviously means that $x$ is not $\nil$).
(in particular, $x$ is not $\nil$).
If this does not hold, then the error state $\bot$ is reached. We define $\bot$ to be a sink state, and any command on $\bot$ transitions to $\bot$ itself. Second, allocation adds a \emph{fresh} location (a new element distinct from all the locations in the current universe) to the universe of locations and to the allocated set $A$. The heap $H$ is also extended, mapping the new location to fixed default values $\mathit{default}_{\!f}$ for each $f \in \Ff_m$. Correspondingly, deallocation removes the deallocated location from $A$. We do not allow double deallocation, and doing so reaches the error state $\bot$. Finally, functions have call-by-value semantics.

\bigskip

\begin{figure}

\begin{mathpar}
\inferrule*[right = error]{~}{\opsem{s}{\bot} \redto \bot }

\inferrule*[right = assn]{\opsem{y}{\sigma} \evalto v}{\opsem{x \coloneq y}{\sigma} \redto \sigma(S)[x\gets v]}

\inferrule*[right = deref]{S(y) \in A \\ \opsem{f(y)}{\sigma} \evalto v}{\opsem{x \coloneq y.f}{\sigma} \redto \sigma(S)[x\gets v]}

\inferrule*[right = badderef]{S(y) \notin A}{\opsem{x \coloneq y.f}{\sigma} \redto \bot}

\inferrule*[right = mut]{ \opsem{y}{\sigma} \evalto v \\ \opsem{x}{\sigma} \evalto w}{\opsem{y.f \coloneq x}{\sigma} \redto \sigma(H)[f \gets f[v \gets w]]]}

\inferrule*[right = badmut]{ \{S(y), S(x)\} \not\subseteq A}{\opsem{y.f \coloneq x}{\sigma} \redto \bot}

\inferrule*[right = alloc]{a \notin A \\ H^{\prime} = \lambda f, x. \ite(x = a, \defaultf, H(f)(x) )}{\opsem{\alloc(x)}{\sigma} \redto (S[x \gets a], H^{\prime}, A \cup \{a\})}

\inferrule*[right = free]{a \in A \\ \opsem{x}{\sigma}\evalto a }{\opsem{\free(x)}{\sigma} \redto \sigma[A \gets A\backslash\{a\}]}

\inferrule*[right = assume]{S \models \eta }{\opsem{\assume(\eta)}{\sigma} \redto \sigma}

\inferrule*[right = if-true]{S \models \eta \\ \opsem{s_1}{\sigma} \redto \sigma_1}{\opsem{\iteprog{\eta}{s_1}{s_2}}{\sigma} \redto \sigma_1}

\inferrule*[right = if-false]{S \not\models \eta \\ \opsem{s_2}{\sigma} \redto \sigma_2}{\opsem{\iteprog{\eta}{s_1}{s_2}}{\sigma} \redto \sigma_2}

\inferrule*[right = seq]{\opsem{s_1}{\sigma} \redto \sigma_1 \\ \opsem{s_2}{\sigma_1} \redto \sigma_2}{\opsem{s_1;s_2}{\sigma} \redto \sigma_2}

\inferrule*[right = call-by-value$_{n, m}$]{\opsem{g}{\sigma_g} \redto \sigma^{\prime}_g \\ 
\sigma_g \preccurlyeq \sigma
\\
\{\sigma^{\prime}(S)(y_1), \dots, \sigma^{\prime}(S)(y_m)\} \subseteq \sigma^{\prime}_g(A)
\\
(\sigma(A)\backslash \sigma_g(A)) \cap \sigma^{\prime}_g(A) = \phi
\\
% \opsem{x_i}{\sigma} \evalto v_i \text{ for each $i \leq n$} \\
% \opsem{y_i}{\sigma^{\prime}_g} \evalto w_i \text{ for each $i \leq m$}
% \\
A^{\prime} = (\sigma(A)\backslash \sigma_g(A)) \cup \sigma^{\prime}_g(A)
\\
S^{\prime} = \lambda x.\ite(S(x) \in \sigma(A)\backslash \sigma_g(A), S(x), \sigma^{\prime}_g(S)(x)) \text{ with domain $\{ x \mid S(x) \in A^{\prime}\}$}
\\
H^{\prime} = \lambda f. ite( f \notin \domain(\sigma_g(H)), \sigma(H)(f), 
\lambda x. \ite(x \notin \sigma^{\prime}_g(A), \sigma(H)(f)(x), \sigma^{\prime}_g(H)(f)(x))
)
% \\
% \text{where $f \in \domain(\sigma(H))$}
}
{\opsem{y_1, \dots, y_m \coloneq g(x_1, \dots, x_n)}{\sigma} \redto (S^{\prime}, H^{\prime}, A^{\prime})}

\end{mathpar}

    \caption{Operational semantics of the programming language. Here, $\sigma = (S,H,A)$ denotes a non-fail state of the program. The error state is denoted by $\bot$. For each sort $\tau$, a 'default' value is fixed. $\defaultf$ denotes the default value of the sort of the codomain of the function. By $\sigma(S)[x \gets v]$, we mean $(S[x \gets v], H, A)$. By $\domain(X)$, we mean the domain of $X$. By $(S_1,H_1,A_1) \preccurlyeq (S_2, H_2, A_2)$, we mean $A_1 \subseteq A_2$ and if $S_1(x) \in A_1$ then $S_2(x) = S_1(x)$ and $H_1(f)(x) = H_2(f)(x)$ for each $f \in \domain(H_1)$. In other words, $(S_1,H_1,A_1)$ is a restriction $(S_2, H_2, A_2)$. Here, assignment of a constant, data field dereference, etc. have not been written as separate rules, as those are similar to those above.}
    \label{fig:opsem}
\end{figure}

%% file: app-fltranslation.tex
\section{Translation of FL to FORD}
\label{app:fl-translation}

Translation of FL to FORD is given in Figure~\ref{fig:fl-translation}. As FL is just FORD extended with the support and cloud operators, the translation becomes straightforward once these have been handled. We denote the translation of $\varphi$ into FORD as $\nabla(\varphi)$.

Recall that$ M \models \Cl{\varphi} \mbox{ iff } M \models \varphi$. The translation to FORD mimics this: $\nabla(\Cl{\varphi}) = \varphi$.

Intuitively, $\nabla(\Sp(\varphi))$ will be the set containing the locations in the support of $\varphi$. For formulas not involving recursive functions, this is straightforward. For example, $\Sp(\key(x) = \key(y))$ would be the locations $\{x, y\}$. The interaction with cloud is as expected: $\Sp(\Cl{\varphi}) = \phi$.

The support of a recursive function will itself be a recursive definition, defined as the lfp of the support of the function.

For example, let us look at the following list definition:
\[
\lst(x) := \ite(x=\nil, \top, (\nxt(x) = \nxt(x)) \wedge \lst(\Cl{\nxt(x)}) \wedge x \notin \Sp(\lst(\Cl{\nxt(x)})))
\]
Notice that the definition refers to the 'support of a list'. We define $\nabla(\Sp(\lst(x)))$ as $\splst(x)$, where 
\begin{gather*}
\splst(x) :=_{\lfp} \nabla(\Sp(\ite(x=\nil, \top, (\nxt(x) = \nxt(x)) \wedge \lst(\Cl{\nxt(x)}) 
\\
\wedge x \notin \Sp(\lst(\Cl{\nxt(x)})))
))
\end{gather*}

Simplifying,

\[
\splst(x) :=_{\lfp} \ite(x=\nil, \phi, \{x\} \cup \splst(\nxt(x)) )
)
\]

This is what we expect intuitively. If $x = \nil$, then the support of the list is empty. Otherwise, it is the head of the list union the support of the rest of the list.

\begin{figure}[h!]
    \makebox[-100pt][r]{$\nabla(\top)$}  $=$  \makebox[0pt][l]{$\top$}
    \\
    \makebox[-100pt][r]{$\nabla(\bot)$} $=$ \makebox[0pt][l]{$\bot$}
    \\
    \makebox[-100pt][r]{$\nabla(c)$} $=$ \makebox[0pt][l]{$c$}
    \\
    \makebox[-100pt][r]{$\nabla(x)$} $=$ \makebox[0pt][l]{$x$}
    \\
    \makebox[-100pt][r]{$\nabla(f(t_{1},\dots,t_{n}))$} $=$ \makebox[0pt][l]{$f(\nabla(t_{1}),\dots,\nabla(t_{n}))$ for any function $f$}
    \\
    \makebox[-100pt][r]{$\nabla(t_{1} = t_{2})$} $=$ \makebox[0pt][l]{$(\nabla(t_{1}) = \nabla(t_{2}))$}
    \\
    \makebox[-100pt][r]{$\nabla(R(t_{1},\dots,t_{n}))$} $=$ \makebox[0pt][l]{$R(\nabla(t_{1}),\dots,\nabla(t_{n}))$ for any relation $R$}
    \\
    \makebox[-100pt][r]{$\nabla(\alpha \wedge \beta)$} $=$ \makebox[0pt][l]{$\nabla(\alpha) \wedge \nabla(\beta)$}
    \\
    \makebox[-100pt][r]{$\nabla(\alpha \vee \beta)$} $=$ \makebox[0pt][l]{$\nabla(\alpha) \vee \nabla(\beta)$}
    \\
    \makebox[-100pt][r]{$\nabla(\neg \varphi)$} $=$ \makebox[0pt][l]{$\neg \nabla(\varphi)$}
    \\
    \makebox[-100pt][r]{$\nabla(ite(\gamma:\alpha, \beta))$} $=$ \makebox[0pt][l]{$ite(\nabla(\gamma):\nabla(\alpha), \nabla(\beta)) $}
    \\
    \makebox[-100pt][r]{$\nabla(ite(\gamma:t_{1}, t_{2})])$} $=$ \makebox[0pt][l]{$ite(\nabla(\gamma):\nabla(t_{1}), \nabla(t_{2})) $}
    \\
    \makebox[-100pt][r]{$\nabla(Sp(\top)) = \nabla(Sp(\bot))$} $=$ \makebox[0pt][l]{$\phi$}
    \\
    \makebox[-100pt][r]{$\nabla(Sp(c))$} $=$ \makebox[0pt][l]{$\nabla(Sp(x)) = \phi$}
    \\
    \makebox[-100pt][r]{$\nabla(Sp(f(t_{1},\dots,t_{n})))$} $=$ \makebox[0pt][l]{$\bigcup_{i =1}^{n} \{t_{i}\} \cup \bigcup_{i = 1}^{n} \nabla(Sp(t_{i}))$ if $f \in \Ff_m$}
    \\
    \makebox[-100pt][r]{$\nabla(Sp(f(t_{1},\dots,t_{n})))$} $=$  \makebox[0pt][l]{$\bigcup_{i = 1}^{n} \nabla(Sp(t_{i}))$ if $f \not\in \Ff_m$}
    \\
    \makebox[-100pt][r]{$\nabla(Sp(Sp(\varphi)))$} $=$ \makebox[0pt][l]{$\nabla(Sp(\varphi))$}
    \\
    \makebox[-100pt][r]{$\nabla(Sp(Sp(t)))$} $=$ \makebox[0pt][l]{$\nabla(Sp(t))$}
    \\
    \makebox[-100pt][r]{$\nabla(Sp(t_{1} = t_{2}))$} $=$ \makebox[0pt][l]{$\nabla(Sp(t_{1})) \cup \nabla(Sp(t_{2}))$}
    \\
    \makebox[-100pt][r]{$\nabla(Sp(R(t_{1},\dots,t_{n})))$} $=$ \makebox[0pt][l]{$\bigcup_{i=1}^{n} \nabla(Sp(t_{i}))$ for $R \in \Rr$}
    \\
    \makebox[-100pt][r]{$\nabla(Sp(\alpha \wedge \beta))$} $=$ \makebox[0pt][l]{$\nabla(Sp(\alpha)) \cup \nabla(Sp(\beta))$}
    \\
    \makebox[-100pt][r]{$\nabla(Sp(\alpha \vee \beta))$} $=$ \makebox[0pt][l]{$\nabla(Sp(\alpha)) \cup \nabla(Sp(\beta))$}
    \\
    \makebox[-100pt][r]{$\nabla(Sp(\neg \varphi))$} $=$ \makebox[0pt][l]{$\nabla(Sp(\varphi))$}
    \\
    \makebox[-100pt][r]{$\nabla(Sp(ite(\gamma:\alpha, \beta)))$} $=$ \makebox[0pt][l]{$\textit{ if } \gamma \textit{ then }  \nabla(Sp(\gamma))\cup \nabla(Sp(\alpha)) \textit{ else } \nabla(Sp(\gamma))\cup \nabla(Sp(\beta))$}
    \\
    \makebox[-100pt][r]{$\nabla(Sp(ite(\gamma:t_{1}, t_{2})))$} $=$ \makebox[0pt][l]{$\textit{ if } \gamma \textit{ then }  \nabla(Sp(\gamma))\cup \nabla(Sp(t_{1})) \textit{ else } \nabla(Sp(\gamma))\cup \nabla(Sp(t_{2}))$}
    \\
    \makebox[-100pt][r]{$\nabla(Sp([\alpha]))$} $=$ \makebox[0pt][l]{$\phi$}
    \\
    \makebox[-100pt][r]{$\nabla(Sp(I(t_{1},\dots,t_{n})))$} $=$ \makebox[0pt][l]{$\bigcup_{i=1}^{n} \nabla(Sp(t_{i})) \cup SpI(t_{1},\dots,t_{n})$
    for $I \in \Ii$,}
    \\
    \makebox[-100pt][r]{}  \makebox[0pt][l]{where $SpI(\bar{x}) \coloneq_{\mathit{lfp}} \nabla(Sp(\rho_{I}(\bar{x})))$ and $I(\bar{x}) \coloneq_{\mathit{lfp}} \rho_{I}(\bar{x})$} 
    \\   
    \makebox[-100pt][r]{$\nabla([\alpha])$} $=$ \makebox[0pt][l]{$\nabla(\alpha)$}
    \caption{Translation from FL to FORD. For an FL formula $\alpha$, the translated FORD formula is $\nabla(\alpha)$, recursively defined above.} 
    \label{fig:fl-translation}
\end{figure}

%% file: app-seplogic.tex
\section{Details for Section~\ref{sec:slfl}: Base Logic}
\label{app:baseseplogic}
\subsection{Proof of Lemma~\ref{lemma51}}

\begin{proof}
%~\\
\begin{enumerate}
\item By structural induction on $\alpha$. Most rules combine supports of smaller formulae, and hence are subsets of $dom(h)$ (including that for disjunction). The rule for $x \carrow{f} y$, 
when the support is defined, $s(x) \in dom(h)$.

 \item 
By structural induction on $\alpha$. 
If $\alpha = \delta$, a heap-independent formula, then $\Supp(\alpha, s, h)=\emptyset = \Supp(\alpha, s, h')$.
If $\alpha = x \carrow{f} y$, note that 
if support of $\alpha$ wrt $h$ is not $\bot$, it must be $S=\{s(x)\}$. Then since $h$ and $h'$ agree on $s(x)$, $s(x)$ is in $dom(h')$ as well, and $\Supp(\alpha,s,h')=\{s(x)\}=S$.\\

If $\alpha = \exists y. (x \carrow{f} y: \alpha')$, notice that $S = \{s(x)\} \cup \Supp(\alpha', s', h)$, where $s'=s[y \mapsto h(f)(s(x))]$.
But since $s(x)$ is in $S$, $dom(h')$ contains it and agrees with pointers on $s(x)$, and hence $h'(f)(s(x))=h(f)(s(x))$. Hence $\Supp(\exists y. (x \carrow{f} y: \alpha,s,h') = \{s(x)\} \cup \Supp(\alpha', s', h') = 
\{s(x)\} \cup \Supp(\alpha', s', h)$ (by induction hypothesis) $= \Supp(\alpha, s, h)$.\\

Let's now consider $\alpha = \textit{ite}(\gamma, \alpha', \beta')$, and assume its support in $h$, $S$, is not $\bot$. Then $\Supp(\gamma, s, h)$ cannot be $\bot$, and since it is a subset of $S$, by the induction hypothesis, $\Supp(\gamma, s, h')=\Supp(\gamma, s, h)$. Note also that
$h \proj \Supp(\gamma, s, h) = h' \proj \Supp(\gamma, s, h')$, and hence $(s, h \proj \Supp(\gamma, s, h)) \models \gamma$ iff $(s, h' \proj \Supp(\gamma, s, h')) \models \gamma $. 
Consider the first case when $(s, h, \Supp(\gamma, s, h)) \models \gamma$ and $\Supp(\alpha, s, h)$ is not $\bot$. 
Then the latter is a subset of $S$, and by induction hypothesis, $\Supp(\alpha, s, h')=\Supp(\alpha, s, h)$, which proves our claim. 
Similarly, in the second case, when $(s, h \proj \Supp(\gamma, s, h)) \not \models \gamma$, $\Supp(\beta, s, h)$ is not $\bot$,  it is a subset of $S$, and by induction hypothesis, $\Supp(\beta, s, h_1)=\Supp(\beta, s, h)$, which also proves our claim.\\

The cases for $\alpha = \alpha' \oplus \beta'$ are trivial, and follow directly from the inductive hypothesis.\\

\item Follows from (2) since $h \proj S$ agrees with $h$ on $S$.

\item Heaplet $h$ agrees with $h_1$ on $dom(h_1)=\Supp(\alpha, s, h_1)$, and
so $\Supp(\alpha,s,h)=dom(h_1)$ (by (2) above). Similarly, heaplet $h$ agrees with $h_2$ on $dom(h_2)=\Supp(\alpha, s, h_2)$, and so $\Supp(\alpha,s,h)=dom(h_2)$.
Hence $dom(h_1)=dom(h_2)$.

\end{enumerate}
\end{proof}

\subsection{Proof of Lemma~\ref{lemma52}}
\begin{enumerate}
 \item  Proved by structural induction on $\alpha$. Base cases are easy. \\

 Consider $\alpha = \exists y. x \carrow{f} y: \alpha'.$ and $(s,h)\models \alpha$. Then
 $s(x) \in dom(h)$, and hence $\Supp(\alpha, s, h) = \{s(x)\} \cup \Supp(\alpha', s', h)$,
 where $s'=s[y \rightarrow h(f)(s(x))]$
 and $(s',h) \models \alpha'$.
 By induction hypothesis $\Supp(\alpha', s', h)=dom(h)$. Hence 
 $\Supp(\alpha, s, h) = dom(h)$.\\

 Consider $\alpha = \alpha_1 \star \alpha_2$,
 and $(s,h) \models \alpha$. 
 Then there are two heaplets $h_1, h_2$ of $h$, 
 the unions of these domains being the domain of $h$, and $(s,h_i) \models \alpha_i$, $i=1,2$.
 It follows by induction hypothesis that
  $\Supp(\alpha_i, s, h_i) = dom(h_i)$.
 By Lemma 4.1, (4) and (3), it follows that
  $\Supp(\alpha_i, s, h) = dom(h_i)$.
 Hence $\Supp(\alpha_1 \star \alpha_2, s, h)$ is
 $\Supp(\alpha_1, s, h) \cup \Supp(\alpha_2, s, h)$ (by definition), which is equal to 
 $dom(h_1) \cup dom(h_2) = dom(h)$.
 The cases for $\alpha_1 \oplus \alpha_2$, where $\oplus \in \{\wedge, \weakconj, \vee\}$ as well as that for $\ite$ expressions are similar.\\
  
 \item Assume there are two subheaplets that satisfy $\alpha$. By (1), the supports of $\alpha$ in each of them would be their own domains. But by Lemma~\ref{lemma51} (4), the heaplets of these domains must be the same. Hence the heaplets are identical. \qed

\end{enumerate}
% \begin{proof}

\section{Details for Section~\ref{sec:slfl}: Extended logic \textit{SL-FL}}
\label{app:seplogic}
The semantics of this logic extends the semantics on the base logic in the following ways. First, we assume that each background sort is a complete lattice (for flat sorts like arithmetic, we can introduce a bottom element and a top element to obtain a complete lattice). Next, we treat the equations defining $\Supp$ in Figure~\ref{fig:slflbsemantics} as recursive definitions with least fixpoint semantics, with
$\subseteq$ over sets as ordering. 

We add support definitions for recursively defined predicates and functions:
$$\Supp(R(\overline{x}), s, h) = \Supp(\rho_R(\overline{x}), s, h)$$
$$\Supp(F(\overline{x}), s, h) = \Supp(\mu_F(\overline{x}), s, h)$$
where 
$R(\overline{x}) = \rho_R(\overline{x})$ 
and
$F(\overline{x}) = \mu_F(\overline{x})$ 
are the definitions of $R$ and $F$, respectively. 
We can show that the least fixpoint of the above equations defining the map $\Supp$ satisfies the properties described in Lemma~\ref{lemma51} for the extended logic as well. 

The semantics of \SLFL\ is now easy to define. 
We modify the semantics given in Figure~\ref{fig:slflbsemantics} as follows, 
%redefining the semantics for $\textit{ite}$ and 
defining the semantics for recursively defined predicates:
\[
\begin{array}{rcl}

(s,h) \models R(\overline{x}) & \textit{~iff~} & 
(s,h) \models \rho(\overline{x}), ~\textit{~where~} R\textit{'s~definition~is~} R(\overline{x})=\rho(\overline{x})
\end{array}
\]

We can prove lemmas analogous to Lemma~\ref{lemma52} for the extended logic.

Finally, we can translate the extended logic to frame logic as well, modifying the translation given in Figure~\ref{fig:sltofl} for $\textit{ite}$
formulas and translating definitions:

\[
\begin{array}{rcll}
\Pi( \textit{ite}(\alpha', \alpha, \beta) ) & = &  \textit{ite}(\Pi(\alpha'), \Pi(\alpha), \Pi(\beta)) \\

\Pi(p(\overline{t})) & = & p(\overline{t}) \\

\Pi(R(\vec{x})) & = & R(\vec{x}) \\

\Pi(x.f)) & = & f(x) \\

\Pi(f(\overline{t})) & = & f(\Pi(\overline{t})) \\

\Pi(\textit{ite}(\alpha, t, t')) &=& \ite(\alpha, \Pi(t), \Pi(t')\\

\Pi(F(\overline{t})) & =& F(\Pi(\overline{t}))\\
\Pi(~R(\vec{x}) :=_\textit{lfp} \rho(\vec{x})~) & = & 
 R(\vec{x}) :=_\textit{lfp} \Pi(\rho(\vec{x}))\\

\Pi(~F(\vec{x}) :=_\textit{lfp} \mu(\vec{x})~) & = & 
 R(\vec{x}) :=_\textit{lfp} \Pi(\mu(\vec{x}))

\end{array}
\]

We can now prove the lemma analogous to Lemma~\ref{lemma53}:

\begin{lemma}\label{lemmaE3}
    Let $g$ be a global heap (a heaplet with domain $Loc$) and $s$ be a store. For frame logic formulas, we interpret the store as an FO interpretation of variables.
    Let $\alpha$ be an \SLFL\ formula.
    Then 
    \begin{itemize}
        \item $g \models_s \Pi(\alpha)$ iff there exists a heaplet $h$ of $g$ such that
        $(s,h) \models \alpha$.
        \item If $g \models_s \Pi(\alpha)$, then
         $\Supp(\alpha, s, g)$ is equal to the value of $\Sp(\Pi(\alpha))$ in $g$. \qed
    \end{itemize}
\end{lemma}

%% file: naturalproofsandlemmas.tex
\section{Validating Verification Conditions}
\label{sec:vcreasoning}

In this section, we describe our validity procedure for checking the quantifier-free frame logic verification conditions generated by the mechanism described in Section~\ref{sec:vcgen}. Our technique has two stages. In the first stage, we translate the quantifier-free FL formulas to quantifier-free FORD formulas in a way that preserves validity. Recall that quantifier-free FL is simply an extension of quantifier-free FORD with the support $\Sp(\cdot)$ and cloud $\Cl{\cdot}$ operators. Intuitively, our translation encodes the semantics of these operators in FORD itself. In the second stage, we reason with the quantifier-free FORD formulas using natural proofs, a mechanism developed in prior work~\cite{qiu13,pek14,loding18}.

\subsection{Stage 1: Translating Quantifier-Free FL to Quantifier-Free FORD}
\label{sec:fl-translation}

The first stage of our reasoning mechanism translates FL formulas to FORD formulas. The key idea is to encode the semantics of the $\Sp$ and $\Cl{\cdot}$ operators within FORD itself. This process introduces new recursive definitions corresponding to the support of existing recursive definitions.

Let us denote the translation by $\nabla$. This function takes as input a quantifier-free FL formula (or term) containing the $\Sp$ and $\Cl{\cdot}$ operators and outputs an FORD formula (resp. term). We describe the formal translation in Appendix~\ref{app:fl-translation} and provide intuition here. First, we obviously have that for any FL formula $\alpha$ that does not mention the support or cloud operators, $\nabla(\alpha) = \alpha$.

\smallskip
Next, to encode the $\Sp$ operator, it turns out that we can essentially mimic the equations describing the semantics of $\Sp$ (Figure~\ref{fig:sp-semantics})! This is easy to see for formulas that do not contain recursively defined symbols. For example, $\nabla(\Sp(f(x)=y)) = \{x\}$ for a mutable function $f \in \Ff_m$, since $\nabla(f(x) = y) = \nabla(f(x)) \cup \nabla(y) = \{x\} \cup \emptyset$. Similarly, $\nabla(\Sp(\alpha \land \beta)) = \nabla(\Sp(\alpha)) \cup \nabla(\Sp(\beta))$. In this way, the $\Sp$ operator can be `compiled away' by simply following the equations in Figure~\ref{fig:sp-semantics}. 

However, this does not work for recursively defined symbols. For example, recall the definition of $\lst(x)$ written using the cloud operator:
\begin{center}
$\lst(x) :=_{\mathit{lfp}} \ite(x=\nil,\, \top,\, \nxt(x) = \nxt(x) \land \lst(\Cl{\nxt(x)}) \land x \notin \Sp(\lst(\Cl{\nxt(x)})))$
\end{center}

Using the recipe described above, the support of $\lst(x)$ yields the following equation:
\begin{align*}
\Sp(\lst(x)) &= \Sp(\ite(x=\nil,\, \top,\, \nxt(x) = \nxt(x) \land \lst(\Cl{\nxt(x)}) \land x \notin \Sp(\lst(\Cl{\nxt(x)}))))\\
&= \ite(x = \nil, \emptyset, \{x\} \cup \Sp(\lst(\Cl{\nxt(x)}))
\end{align*}

\noindent
where we have expanded $\Sp$ on a cloud expression using the equation $\Sp(\Cl{\alpha}) = \emptyset$.

We cannot compile away the $\Sp$ operator in this case. However, recall that the $\Sp$ operator is the least interpretation satisfying the equations in Figure~\ref{fig:sp-semantics}. 
%cannot be compiled away anymore, but notice that the above equation looks very similar to a recursive definition of the heaplet of a linked list! This is no coincidence, and it turns out that we can 
We therefore introduce a new recursively defined symbol $\Sp_\lst$, defined by
\begin{center}
$\Sp_\lst(x) :=_{\mathit{lfp}} \ite(x = \nil, \emptyset, \{x\} \cup \Sp_\lst(\Cl{\nxt(x)}))$
\end{center}

\noindent
and we then translate $\Sp(\lst(x))$ to $\Sp_\lst(x)$. The latter is now an expression in FORD mentioning the newly created recursive definition. In general, our translation creates a new recursively defined symbol $\Sp_I$ for every $I \in \Ii$ corresponding to its support.

\smallskip
Finally, the encoding of the $\Cl{\cdot}$ operator is trivial. Since support expressions were eliminated in the previous step, $\Cl{\alpha}$ is identical to $\alpha$. Therefore, we simply ignore it in the translation. In particular, the recursive sub-expression $\Sp_\lst(\Cl{\nxt(x)})$ in the above definition becomes $\Sp_\lst(\nxt(x))$ in the final translation to FORD.

\subsection{Stage 2: Reasoning with FORD Formulas using Natural Proofs and SMT Solvers}
\label{sec:natproofs}

The above translation of frame logic formulas results in FORD formulas that are quantifier-free (i.e., all free variables are implicitly universally quantified). Furthermore, they have the special form that functions from the location sort map to either the location sort or a background sort, but no function maps background sorts to the location sort. Hence they are part of a special fragment with ``one-way'' definitions, $\mathcal{L}_\textit{oneway}$ fragment.

Natural proofs~\cite{qiu13,pek14,loding18} are a sound but incomplete technique for proving validity of FORD formulae. It works by treating recursive definitions to have fixpoint semantics (rather than least fixpoint) and hence obtaining a first-order formula. It then instantiates the recursive definitions (similar to instantiating the quantifier for these definitions) on terms of depth $d$ (for larger and larger $d$), and obtaining quantifier-free formulas. Validating these quantifier-free formulas is performed using an  SMT solver. The foundations of this technique~\cite{loding18}
show that the technique is in fact a complete technique for dealing with first-order formulae (where recursive definitions have fixpoint semantics). Furthermore, in practice, natural proofs have been shown useful for heap verification of other logics.
We can hence validate verification conditions using this automated technique.

%% file: app-eval.tex
\section{Details of Evaluation}
\label{app:eval}

\subsection{Benchmark Suite}
Table~\ref{tab:eval-benchmarks} summarizes our benchmark suite. We report in addition to lines of code and annotations,  (specification, contracts, loop invariants, etc.) and the number of lemmas needed for each data structure class.

\begin{table}[ht]\scriptsize
    \centering
    \begin{tabular}{cc|ccccc}\hline
    
        Data Structure &Operations & Lines of &Lines of&Rec. Defs. & No. of unique&Lemmas\\
        & &Code+Annot. (\textit{FL}) &Code +Annot.(\SLFL\ ) && rec def heaplets& \\ \hline

        \multirow{10}{*}{Singly Linked List}
            &Append &30 &29 &2 &1 & 0 \\
            &Copy All &32 &33 &2 &1 & 0 \\
            &Delete      &35 &35 &2&1& 0 \\
            &Find        &34 &35 &2&1& 0 \\
            &Insert Back &32 &32 &2&1& 0 \\
            &Insert Front &21 &21 &2&1& 0 \\
            &Reverse     &33 &33 &2&1& 0 \\

            &Insertion Sort &60 &62 &4 &1 & 1 \\
            &Merge Sort &107 &110 &4 &1 & 1 \\
            &Quicksort  &145 &152 &5 &1 & 1 \\ \hline
            
        \multirow{4}{*}{Doubly Linked List}
            &Insert Back &38 &44 &2 &1 &\multirow{4}{*}{0} \\
            &Insert Front   &36 &41 &2 &1 &  \\
            &Delete Mid &50 &65 &4
            &2 & \\
            &Insert Mid &56 &77 &4 &2 &  \\ \hline
            
        \multirow{3}{*}{Sorted List}
            &Delete &40 &42 &4&1 &\multirow{3}{*}{1} \\
            &Find   &44 &47 &4 &1 & \\
            &Insert &44 &46 &3 &1 & \\
            \hline
            
        \multirow{3}{*}{Circular List}
            &Delete Front &61 &61 &4 &2 &\multirow{3}{*}{1} \\
            &Find &61 &62 &4 &2 & \\
            &Insert Front &42 &42 &4 &2 &  \\ \hline
            
        \multirow{4}{*}{BST}
            &Delete &113 &120 &4 &1 &\multirow{5}{*}{2} \\
            &Find      &76 &76 &4 &1 & \\
            &Insert    &86 &88 &4 &1 & \\
            &Rotate Right &58  &59 &4 &1 & \\ 
            &BST to List &94 &90 &7 &2 & \\ \hline

        \multirow{1}{*}{Tree}
            &In-order Traversal &65 &65 &3 &1 &0\\ \hline 
        \multirow{2}{*}{Treap}
            &Delete &154 &158 &\multirow{2}{*}{6} &\multirow{2}{*}{1} &4 \\
            &Find   &95 &97 & & &2 \\ \hline
            
        \multirow{1}{*}{RBT}
            &Insert &193 &218 &5 &1 &2 \\\hline
    \end{tabular}
    \caption{List of benchmarks, the lines of code, the number of recursive definitions and the number of these with different heaps, and the number of lemmas needed to prove the benchmark.}
    \label{tab:eval-benchmarks}
\end{table}

Table~\ref{tab:eval-contracts} contains examples of pre- and post-conditions for sample routines. \textit{FL} specifications use the $\Sp$ operator to express disjointness of data structures. We also note that specifications may need to be strengthened for verification. For example, BST Remove Root requires the following constraint in the post-condition  $ (\mathit{Old}(\mathit{Min}(x)) \leq \mathit{Min}(ret))\wedge (\mathit{Max}(ret) \leq \mathit{Old}(\mathit{Max}(x)))$ to verify. 

\begin{table}[ht]\footnotesize
    \centering
    \begin{tabular}{l l l} \hline
        Benchmark &Pre Condition & Post Condition \\ \hline

        Sorted Insert 
        &$\mathit{Sorted}(x)$ 
        &\makecell[l]{$\mathit{Sorted}(\mathit{ret}) \land \mathit{Keys}(\mathit{ret}) = \mathit{Old}(\mathit{Keys}(x)) \cup \{ k \}$ 
        \\ $\land$ $\ite(\mathit{Old}(\mathit{Min}(x)) < k,$ $\mathit{Min}(\mathit{ret}) = \mathit{Old}(\mathit{Min}(x)),$ 
        $\mathit{Min}(\mathit{ret}) = k)$} \\\hline
        
        \makecell[l]{Sorted Concat \\ (part of Quicksort)} &\makecell[l]{$\mathit{Sorted}(x) \land$ $\mathit{Sorted}(y)$ 
        \\$\land$ $\mathit{Max}(x) \leq \mathit{Min}(y)$ 
        \\$\land$ $(Sp(\mathit{Sorted}(x))$ 
        \\ \hspace{0.2cm} $\cap$ $Sp(\mathit{Sorted}(y))) = \emptyset$} 
        &\makecell[l]{ $\mathit{Sorted}(\mathit{ret})$ $\land$ $\mathit{Keys}(\mathit{ret}) = \mathit{Old}(\mathit{Keys}(x)) \cup \mathit{Old}(\mathit{Keys}(y))$ 
        \\ $\land$ $\ite(x = \nil, \mathit{Min}(\mathit{ret}) = \mathit{Old}(\mathit{Min}(y)),$ $\mathit{Min}(\mathit{ret}) = \mathit{Old}(\mathit{Min}(x)))$ 
        \\$\land$ $\ite(y = \nil, \mathit{Max}(\mathit{ret}) = \mathit{Old}(\mathit{Max}(x)),$ 
        $\mathit{Max}(\mathit{ret}) = \mathit{Old}(\mathit{Max}(y)))$} \\\hline

    \end{tabular}
    (a) Example \textit{FL} Pre and Postconditions 
    \vspace{0.1cm}

    \begin{tabular}{l l l} \hline
        % SLFL benchmarks
        Benchmark &Pre Condition & Post Condition \\ \hline

        BST Rotate Right
        &\makecell[l]{$(x \not= \nil)$\\ 
        $\star (\mathit{BST}(x) \weakconj (\mathit{left}(x) \not= \nil ))$} 
        
        &\makecell[l]{$\mathit{BST}(ret) \wedge 
          (\mathit{Keys}(ret) = \mathit{Old}(\mathit{Keys}(x)))$
          } 
          \\\hline

        \makecell[l]{BST Remove Root\\ (part of BST Delete)}
        &\makecell[l]{$(x \not= \nil) \star (\exists \mathit{lft}. (x \carrow{\mathit{left}} \mathit{lft}:$
        \\
         $( \exists \mathit{rht}. (x \carrow{\mathit{right}} \mathit{rht}:
            (
              (k = \mathit{key}(x))$\\ 
              $\star 
              (\mathit{BST}(\mathit{lft}) 
              \wedge (\mathit{Max}(\mathit{lft}) < k))$\\
              $\star
              (\mathit{BST}(rht) 
              \wedge (k < \mathit{Min}(\mathit{rht})))
              )
            )
         )
        )$} 
        
        &\makecell[l]{$\mathit{BST}(ret) \wedge 
          (\mathit{Keys}(ret) = (\mathit{Old}(\mathit{Keys}(x)))\setminus \{k\})$\\
          $\wedge 
          (\mathit{Old}(\mathit{Min}(x)) \leq \mathit{Min}(ret))
          \wedge 
          (\mathit{Max}(ret) \leq \mathit{Old}(\mathit{Max}(x)))
          $} 
          \\\hline

    \end{tabular}
    (b) Example \SLFL\ Pre and Post conditions 
    \vspace{0.05cm}
    
    \caption{Example Pre and Post Conditions. $x$ and $y$ are of type $\Loc$, and $k$ is of type $\mathit{Int}$. $\mathit{ret}$ is a pointer to the data structure returned by the routines. 
    }
    \label{tab:eval-contracts}
\end{table}

Table~\ref{tab:eval-lemmas} contains some examples of lemmas used in our benchmarks.

\begin{table}[ht]\footnotesize
    \centering
    \begin{tabular}{cc} \hline
        Benchmark &Lemmas \\\hline

        Insertion Sort &$\mathit{Sorted}(x) \implies \mathit{List}(x)$ 
        
        \\\hline
        
        BST Find &$\mathit{BST}(x) \implies \mathit{Max}(x) < k \implies k \not\in \mathit{Keys}(x)$ \\\hline
        
        Tree to List &\makecell{$\mathit{BST}(x) \implies n \neq \mathit{nil} \implies$ \\ $\ite(\mathit{left}(x) = \mathit{nil}, \mathit{Min}(x) = key(x), \mathit{Min}(x) = \mathit{Min}(\mathit{left}(x)))$} 
        
        \\\hline

        Quicksort &$x \neq \mathit{nil} \implies \mathit{Max}(x) \leq \mathit{Min}(y) \implies \mathit{Min}(x) \leq \mathit{Min}(y)$ \\\hline

    \end{tabular}
    \caption{Example Lemmas}
    \label{tab:eval-lemmas}
\end{table}

\subsection{Definitions of Datastructures}
%\subsubsection{Definitions of Datastructures}
The definitions of the FL data structures used in the benchmarks 
%are similar to the definitions in the Dryad repository, though the Frame Logic definitions 
use explicit support and cloud operators.
The definition for singly-linked lists is as given below:
\begin{equation*}
    \mathit{List(x)} := \ite(x = \nil, \top, \mathit{List}(\nxt(x)) \land x \not\in \Sp(\mathit{List}([\nxt(x)]))
\end{equation*}
and the definition for doubly-linked lists is similar, with the addition of a constraint that if $\nxt(x)$ is not nil then its previous pointer should point back to $x$,
\begin{align*}
    \mathit{Dll}(x) := \ite(&x = \nil, \top, \ite(\nxt(x) = \nil, \top, \\ &x = \mathit{prev}(\nxt(x)) \land \mathit{Dll}(\nxt(x)) \land x \not\in \Sp(\mathit{Dll}([\nxt(x)]))))
\end{align*}
Next, to define sorted lists, we rely on an auxiliary definition of the minimum element in a list and then a (non-decreasing) sorted list is simply a list where the front key is no greater than the minimum value of its tail, which is itself also a sorted list:
\begin{align*}
    \mathit{Min}(x) := \ite(&x = \nil, +\infty, \ite(\key(x) < \mathit{Min}(\nxt(x)), \key(x), \mathit{Min}(\nxt(x)))) \\
    \mathit{Sorted}(x) := \ite(&x = \nil, \top, \\ &\mathit{Sorted}(\nxt(x)) \land x \not\in \Sp(\mathit{Sorted}([\nxt(x)])) \land \key(x) \leq \mathit{Min}(\nxt(x)))
\end{align*}
In our implementation of this definition, $+\infty$ is simply an additional variable that is treated as having an infinitely large value.

For circular lists, %like the Dryad repository, 
we define list segments using the $\mathit{Lseg}(\cdot)$ definition and we then define a circular list as empty or reaching itself through its next pointer,
\begin{align*}
    \mathit{Lseg}(x, y) := \ite(&x = y, \top, \ite(x = \nil, \bot, \\ &\mathit{Lseg}(\nxt(x), y) \land x \not\in \Sp(\mathit{Lseg}([\nxt(x)], y)))) \\
    \mathit{Circ}(x) := \ite(&x = \nil, \top, \mathit{Lseg}(\nxt(x), x) \land x \not\in \Sp(\mathit{Lseg}([\nxt(x)], x)))
\end{align*}
Our definition for binary search trees, like for sorted lists, relies on a definition of the minimum and maximum values in a tree and the binary search property is specified by requiring the maximum element in the left subtree to be less than the key of the root and the minimum element in the right subtree to be greater than the key of the root.
The definitions of $\mathit{Min}$ and $\mathit{Max}$ are similar to that for sorted lists and so are omitted; we note that the default value for $\mathit{Max}$ if $x$ is $\nil$ is $-\infty$.
Then, the definition of a binary search tree is 
\begin{align*}
    \mathit{BST}(x) := \ite(&x = \nil, \top, \\ &\mathit{BST}(\mathit{left}(x)) \land \mathit{Max}(\mathit{left}(x)) < \mathit{key}(x) \land \mathit{BST}(\mathit{right}(x)) \\ &\land \key(x) < \mathit{Min}(\mathit{right}(x)) \land x \not\in \Sp(\mathit{BST}([\mathit{left}(x)])) \\ &\land x \not\in \Sp(\mathit{BST}([\mathit{right}(x)])) \land \Sp(\mathit{BST}([\mathit{left}(x)])) \cap \Sp(\mathit{BST}([\mathit{right}(x)])) = \emptyset)
\end{align*}
In addition to our normal requirement that $x$ not be in the support of either of the subtrees we also specify here that the supports of the subtrees must be disjoint; together these properties express that $x$ is indeed a tree.

The definitions for Treaps and Red-Black Trees are similar to this definition for binary search trees, augmented with their appropriate additional constraints, the max heap property on the priorities in the Treap and the red-black property for Red-Black trees.

The implementation of the max heap property for the Treap is in the same fashion as the binary search property for binary-search trees, except that we just require that the priority of the root be greater than the maximum priority of either subtree.

For red-black trees we add a recursive definition which determines the (maximum) black-height of a tree and then a tree is a red-black tree if both subtrees have the same black height and the root is black or both children are black.

\subsection{An Illustrative Example}

In this section, we step through the VC generation of a program to illustrate the transformation from an \SLFL\ annotated program to FL VCs.

Given below is a program RevPrepend(x, y), which takes two linked lists and prepends the reverse of the first onto the second. The program has two recursive functions, $\lst$ and $\Keys$, written in \SLFL. Recall from Section~\ref{sec:eval} that the support of $\Old(.)$ is always empty.

\begin{lstlisting}
List(x) := ite( x = nil, True, Exists y:  next(x). (next(x) = next(x)) * List(y))
Keys(x) := ite( x = nil, EmptySetInt, SetAdd(Keys(next(x)), x))

RevPrepend(x, y) returns(ret)
@pre: List(x) * List(y)
@post: List(ret) and (Keys(ret) = SetUnion( Old(Keys(x)), Old(Keys(y))))
(if (x = nil)
 then
   ret := y;
 else
   tmp := x.next;
   x.next := y;
   ret := RevPrepend(tmp, x);
)
return;
\end{lstlisting}

First, we convert the recursive definitions and the program annotations into FL (see Fig \ref{fig:sltofl} for the translation). For instance, $\lst$ becomes
\begin{gather*}
\lst(x) := \ite(x=\nil, \top, \exists y : y = \nxt(x).\, (\nxt(x) = \nxt(x)) \land \lst(y)
\\
\land (\Sp(\nxt(x) = \nxt(x)) \cap \Sp(\lst(y))) = \phi
)
\end{gather*}
Using Lemma~\ref{lem:cloud-lemma}, we eliminate the existential quantifier,
\begin{gather*}
\lst(x) := \ite(x=\nil, \top, (\nxt(x) = \nxt(x)) \land (\nxt(x) = \nxt(x)) \land \lst(\Cl{\nxt(x)})
\\
\land (\Sp(\nxt(x) = \nxt(x)) \cap \Sp(\lst(\Cl{\nxt(x)}))) = \phi
)
\end{gather*}

Observe that this mechanical translation introduces a blow-up in the size of the formula. For example, it introduces a check if the supports of $\nxt(x) = \nxt(x)$
and $\Sp(\lst(\Cl{\nxt(x)}))$ are disjoint. This can be simplified to just checking if $x$ is in the support of the latter. A more drastic example of the translation introducing superfluous checks can be seen in the postcondition, which gets translated to 
\begin{gather*}
    \lst(ret) \wedge (\Keys(\ret) = \Old(\Keys(x)) \cup \Old(\Keys(y))) \wedge \\
    \Sp(\lst(ret)) = \Sp(\Keys(\ret) = \Old(\Keys(x)) \cup \Old(\Keys(y)))
\end{gather*}

One can check that the last conjunct always holds (the support of $\lst$ and that of $\Keys$ is always the same), and that the support of the entire expression is equal to the support of $\lst(ret) \wedge (\Keys(\ret) = \Old(\Keys(x)) \cup \Old(\Keys(y)))$. Thus, the post condition is equivalent to $\lst(ret) \wedge (\Keys(\ret) = \Old(\Keys(x)) \cup \Old(\Keys(y)))$. Our optimizations involve performing such simplifications syntactically, reducing the burden on the SMT solver. The following is an FL-annotated version of RevPrepend.

\begin{lstlisting}
List(x) := ite( x = nil, True, (next(x) = next(x)) and List([next(x)]) and Not(IsMember(x, Sp(List([next(x)])))))
Keys(x) := ite( x = nil, EmptySetInt, SetAdd(Keys(next(x)), x))

RevPrepend(x, y) returns(ret)
@pre: List(x) and List(y) and IsEmpty( SetIntersect( Sp(List(x)), Sp(List(y))))
@post: List(ret) and (Keys(ret) = SetUnion( Old(Keys(x)), Old(Keys(y)))) 

(if (x = nil)
 then
   ret := y;
 else
   tmp := x.next;
   x.next := y;
   ret := RevPrepend(tmp, x);
)
return;
\end{lstlisting}

We shall not go over the entire VC generation in detail; we will just look at the interesting cases involving the function call. The other BBs are straightforward. For example, the 'if' case simply sets $T$ to be the conjunct of the precondition, $x = \nil$, and $\ret= y$. Finally, it checks if the post-condition holds. Since $\ret = y$, this follows directly from the pre-condition. In the BBs checking dereference safety, say, of $tmp := x.\nxt$, $T$ will be the precondition and $x \not = \nil$, and the VC will check if $x \in A$ (where $A$ is just the support of the precondition).

When calling RevPrepend(tmp, x), we need to ensure memory safety - that is, the locations initially accessed by the function call is contained \emph{within} the allocated set, motivating the need for relaxed post-conditions. This is given by the following BB.

\begin{lstlisting}
{List(x) and List(y) and IsEmpty( SetIntersect( Sp(List(x)), Sp(List(y))))} 
# T =  List(x) and List(y) and IsEmpty( SetIntersect( Sp(List(x)), Sp(List(y)))), 
# A = Sp(T), Fr = {}, RD = {List, Keys}, H = \lambda f. f where f \in {next, key}
assume( x != nil);
# T <-  T and x != nil
tmp := x.next;
# T <- T and tmp = next(x)
x.next := y;
# H1 = H[next <- lambda arg. ite(arg = x, y, next(arg))]
# List1 = List[H1], Keys1 = Keys[H1], RD1 = {List, Keys, List1, Keys1}
# Fr1 = {fr({x}, H, H1)}
{RP: List1(tmp) and List1(x) and IsEmpty( SetIntersect( Sp(List1(tmp)), Sp(List1(x))))}
\end{lstlisting}
The VC generated for this BB is $(\fr({x}, H, H1) \wedge T) \implies ( \varphi \wedge \Sp(\varphi) \subseteq A )$, where $\varphi = (\lst1(tmp) \wedge \lst1(x) \wedge 
(\Sp(\lst1(\mathit{tmp})) \cap \Sp(\lst1(\mathit{x})) = \phi))$.

To ensure tightness of the heap, we must characterize the heap after the function call. Inductively, we know the 
portion of the heap modified by the function call is characterized by the postcondition of RevPrepend(tmp, x). This motivates the function call rule. 

\begin{lstlisting}
{List(x) and List(y) and IsEmpty( SetIntersect( Sp(List(x)), Sp(List(y))))} 
assume( x != nil)
tmp := x.next
x.next := y
ret := RevPrepend(tmp, x)
# Let PreCall = (List1(tmp) and List1(x) and IsEmpty( SetIntersect( Sp(List1(tmp)), Sp(List1(x)))))
# H2[next] = lambda arg. ite(Not(IsMember(arg, Sp(PreCall))), H1(next)(arg), next2(arg))
# H2[key] = lambda arg. ite(Not(IsMember(arg, Sp(PreCall))), H1(key)(arg), key2(arg))
# List2 = List1[H2], Keys2 = Keys1[H2]
# Let PostCall = (List2(ret) and (Keys2(ret) = SetUnion( Old(Keys1(tmp)), Old(Keys1(x)))) )
# T2 = T and PostCall
# A2 = (A\Sp(PreCall)) U Sp(PostCall)
# Fr2 = Fr1 U {fr(Sp(PreCall), H1, H2)}
# RD2 = RD1 U {List2, Keys2}
{List2(ret) and (Keys2(ret) = SetUnion( Old(Keys(x)), Old(Keys(y))))}
\end{lstlisting}
The VC generated for this BB is $(\fr({x}, H, H1) \wedge \fr(\Sp(\text{PreCall}), H1, H2) \wedge T2) \implies \varphi \wedge \Sp(\varphi) = A2$, where $\varphi = \lst2(\ret) \wedge (\Keys2(\ret) = (\Old(\Keys(x)) \cup \Old(\Keys(y))))$.

\medskip

We now describe a high-level proof of the validity of the VC. Recall that the contract of RevPrepend tells us if we start with two disjoint lists and the program is executed - (1) the return is a list whose keys are the union of the keys of the original two lists, and (2) the heap support of the postcondition gives the 'modified' heap of the program.

From the validity of the previous BB, the precondition of $RevPrepend(\tmp, x)$ holds for the function call. Thus, we can conclude that after the function call - (1) PostCall = $\lst2(\ret) \wedge (\Keys2(\ret) = \Keys1(x) \cup \Keys1(\tmp))$ holds, and (2) $\Sp(\text{PostCall})$ is the modified heap of the call.

We now prove $\varphi$. It suffices to show $\Keys1(x) \cup \Keys1(\tmp) = \Keys(x) \cup \Keys(y)$. To do so, observe that only $x$ is modified before the function call. Thus, by the frame rule Fr1, $\Keys1(y) = \Keys(y)$ and $\Keys(\tmp) = \Keys1(\nxt(x))$. The result follows. 

The proof of $\Sp(\varphi) = A2$ is similar. By Fr1, it follows have $\Sp(\lst1(y)) = \Sp(\lst(y))$ and $\Sp(\lst1(\tmp)) = \Sp(\lst(\tmp))$. It follows that $\Sp(\text{PreCall}) = A$. Thus, $A2 = \Sp(\text{PostCall}) = \Sp(\varphi)$.